\documentclass[runningheads,autoref]{lipics-v2021}
\usepackage[T1]{fontenc}
\usepackage{graphicx}
\usepackage{multirow}
\usepackage{array}
\usepackage{subcaption}
\usepackage{xspace}
\usepackage{listings}
\usepackage{hyperref}
\usepackage{amssymb,amsmath}
\usepackage{stmaryrd}
\usepackage[nolist,nohyperlinks]{acronym}
\usepackage{tabularx,booktabs}
\usepackage{algorithm}
\usepackage[noend]{algpseudocode}
\usepackage{bm,enumitem}
\usepackage{cleveref}
\usepackage{tikz}
\usetikzlibrary{calc}
\usetikzlibrary{decorations.pathmorphing}

\usepackage{todonotes}

\usepackage{xcolor}

\definecolor{Green}{rgb}{0.09, 0.45, 0.27}

\newcommand{\blue}[1] {%
	\protect\leavevmode%
  \begingroup%
	\color{blue}%
    #1%
	\endgroup%
}

\newcommand{\green}[1] {%
	\protect\leavevmode%
  \begingroup%
	\color{Green}%
    #1%
	\endgroup%
}

\newcommand{\magenta}[1] {%
	\protect\leavevmode%
  \begingroup%
	\color{magenta}%
    #1%
	\endgroup%
}
\newcommand{\orange}[1] {%
	\protect\leavevmode%
  \begingroup%
	\color{orange}%
    #1%
	\endgroup%
}

\newtheorem{innercustomlemma}{Lemma}
\newenvironment{dupllemma}[1]
  {\renewcommand\theinnercustomlemma{#1}\innercustomlemma}
  {\endinnercustomlemma}

\newtheorem{innercustomthm}{Theorem}
\newenvironment{dupltheorem}[1]
  {\renewcommand\theinnercustomthm{#1}\innercustomthm}
  {\endinnercustomthm}

\newtheorem{invariant}{Invariant}
\crefname{invariant}{Invariant}{Invariants}

\newcommand*{\verit}{\texttt{veriT}}

\newcommand*{\trail}{\tau}
\newcommand*{\q}{\omega}
\newcommand*{\p}{\pi}

\newcommand*{\wl}{\mathrm{WL}} 
\newcommand*{\reason}{\rho} 
\newcommand*{\level}{\delta}
\newcommand*{\chunks}{\gamma}
\newcommand*{\crossChunks}{\eta}
\newcommand*{\weight}{\zeta}
\newcommand*{\cadical}[0]{\textsc{CaDiCaL}\xspace} %

\newcommand{\Continue}{\textbf{continue} }

\newcommand*{\napsat}[0]{\textsc{NapSAT}\xspace} 
\newcommand*{\minisat}[0]{\textsc{MiniSAT}\xspace} %
\newcommand*{\vampire}[0]{\textsc{Vampire}\xspace} %
\newcommand*{\cvc}[0]{\textsc{cvc5}\xspace} %

\tikzstyle{vertex}=[draw,minimum size=20pt,inner sep=1pt]
\tikzstyle{implied}=[circle]
\tikzstyle{decision}=[rectangle]
\tikzstyle{myarr}=[shorten >=1pt,->,>=stealth]
\tikzstyle{marked}=[fill=blue!15]
\tikzstyle{conflict}=[dashed, thick]

\tikzset{
  prefix after node/.style={prefix after command=\pgfextra{#1}},
  /semifill/ang/.initial=45,
  /semifill/upper/.initial=none,
  /semifill/lower/.initial=none,
  semifill/.style={
    circle, draw,
    prefix after node={
      \pgfqkeys{/semifill}{#1}
      \path let \p1 = (\tikzlastnode.north), \p2 = (\tikzlastnode.center),
                \n1 = {\y1-\y2} in [radius=\n1]
            (\tikzlastnode.\pgfkeysvalueof{/semifill/ang})
            edge[
              draw=none,
              fill=\pgfkeysvalueof{/semifill/upper},
              to path={
                arc[start angle=\pgfkeysvalueof{/semifill/ang}, delta angle=180]
                -- cycle}] ()
            (\tikzlastnode.\pgfkeysvalueof{/semifill/ang})
            edge[
              draw=none,
              fill=\pgfkeysvalueof{/semifill/lower},
              to path={
                arc[start angle=\pgfkeysvalueof{/semifill/ang}, delta angle=-180]
                -- cycle}] ();}}}

\tikzset{
  node split radius/.initial=1,
  node split color 1/.initial=magenta!25,
  node split color 2/.initial=orange!25,
  node split color 3/.initial=blue!25,
  node split half/.style={node split={#1,#1+180}},
  node split/.style args={#1,#2}{
    path picture={
      \tikzset{
        x=($(path picture bounding box.east)-(path picture bounding box.center)$),
        y=($(path picture bounding box.north)-(path picture bounding box.center)$),
        radius=\pgfkeysvalueof{/tikz/node split radius}}
      \foreach \ang[count=\iAng, remember=\ang as \prevAng (initially #1)] in {#2,360+#1}
        \fill[line join=round, fill=\pgfkeysvalueof{/tikz/node split color \iAng}]
          (path picture bounding box.center)
          --++(\prevAng:\pgfkeysvalueof{/tikz/node split radius})
          arc[start angle=\prevAng, end angle=\ang] --cycle;
} } }

\tikzstyle{ck_a}=[fill=green!25]
\tikzstyle{ck_b}=[fill=orange!25]
\tikzstyle{ck_c}=[fill=blue!25]
\tikzstyle{ck_d}=[fill=magenta!25]
\tikzstyle{ck_ab}=[semifill={upper=orange!25, lower=green!25}]
\tikzstyle{ck_cd}=[semifill={upper=magenta!25, lower=blue!25}]
\tikzstyle{ck_bcd}=[node split={45, 165, 285}]

\DeclareMathOperator*{\argmin}{arg\,min}

\newcolumntype{P}[1]{>{\centering\arraybackslash}p{#1}}

\renewcommand{\paragraph}[1]{\par\smallskip\noindent\textbf{#1}}

\title{Generalizing CDCL with Graph Backtracking}

\author{%
	Robin Coutelier}
	{TU Wien, Vienna, Austria}
	{robin.coutelier@tuwien.ac.at}
	{0009-0002-4735-5215}
	{}

\author{%
    Thomas Hader}
	{TU Wien, Vienna, Austria}
	{thomas.hader@tuwien.ac.at}
	{0009-0002-8920-5469}
	{}

\author{%
	Laura Kovács}
	{TU Wien, Vienna, Austria}
	{laura.kovacs@tuwien.ac.at}
	{0000-0002-8299-2714}
	{}

\authorrunning{Coutelier et al.}

\Copyright{Robin Coutelier, Thomas Hader, and Laura Kovács}


\EventEditors{Alexey Ignatiev and Stefan Szeider}
\EventNoEds{2}
\EventLongTitle{29th International Conference on Theory and Applications of Satisfiability Testing (SAT 2026)}
\EventShortTitle{SAT 2026}
\EventAcronym{SAT}
\EventYear{2026}
\EventDate{July 20--23, 2026}
\EventLocation{Lisbon, Portugal}
\EventLogo{}
\SeriesVolume{377}
\ArticleNo{1}

\acknowledgements{We would like to thank Mathias Fleury, Michael Rawson, Clemens Eisenhofer, Márton Hajdu, and Pascal Fontaine for fruitful discussions. We also thank the anonymous reviewers for their feedback.}
\funding{The authors acknowledge support from the ERC Consolidator Grant ARTIST 101002685; the TU Wien Doctoral Colleges TrustACPS and SecInt; the FWF SpyCoDe SFB projects F8504; and the WWTF Grant ForSmart 10.47379/ICT22007.}

\nolinenumbers
\begin{document}
	\setcounter{footnote}{0}
	\maketitle
	\begin{abstract}
		We present \emph{graph backtracking}, a novel, fine-grained backtracking scheme for CDCL-based SAT solving, parametrized by a user-defined weight function.
		For conflict repair, we challenge the decision level abstraction and use the implication graph as a precise guiding structure to minimize the weight of literals that are unassigned.
		Graph backtracking is sound, complete, and
                terminating. We show that it is a generalization of
                chronological and non-chronological backtracking by
                simulating them with specific weight functions.
		Our approach is implemented in the experimental solver
                \napsat. Empirical results show that graph
                backtracking requires fewer literal propagations than
                standard approaches, leading to improved solver runtime.
	\end{abstract}

	\acresetall



\section{Introduction}
The CDCL algorithm~\cite{DBLP:journals/tc/Marques-SilvaS99} is the
dominant approach in propositional satisfiability (SAT)
solving~\cite{DBLP:series/faia/336}.
CDCL solvers typically employ a rather aggressive backtracking scheme
for conflict repair, referred to as non-chronological backtracking
(NCB).
More recently, chronological backtracking (CB)
has been shown to complement NCB in SAT solving~\cite{DBLP:conf/sat/NadelR18,DBLP:conf/sat/MohleB19},
with dedicated CB-based algorithms improving the state of the art in
model counting~\cite{DBLP:conf/sat/MohleB19,DBLP:conf/cav/SoosM25} and
AllSAT~\cite{DBLP:journals/ai/SpallittaSB25},
with further applications to MaxSAT~\cite{DBLP:conf/sat/Nadel22}.
However, existing backtracking schemes maintain a strict top-down order: undoing the most recent decisions
and their consequences first.
This approach provides strong invariants and enables efficient
implementations,
but restricts the locality of the search and may lead to unnecessary
backtracking.
Orthogonal to backtracking, user propagators~\cite{DBLP:conf/vmcai/BjornerEK23,DBLP:journals/jair/FazekasNPKSB24} give users more control over the behavior of the SAT solver by allowing them to extract partial assignments, raise application-specific conflicts, assign new literals, and provide new clauses during proof search.

Inspired by developments in user propagators and backtracking strategies, we introduce \emph{graph backtracking} (GB), a more
general backtracking scheme for CDCL-based SAT solving. GB enables a more flexible and fine-grained mechanism for repairing conflicts.
It determines precisely which literals can be unassigned and performs minimal backtracking.
Furthermore, instead of always undoing the latest assigned variables, a user can supply a weight
function, allowing more control over backtracking.
In particular, the user may prefer some
literals to remain assigned longer. Associating them with a higher weight
discourages the solver from unassigning them, thus achieving the
desired behavior.


\paragraph{\bf Motivating example.} We motivate graph backtracking via a
SAT-based search. Let $h_1$ and $h_2$ denote comparatively heavy literals, and the following clause set $F$:
\begin{equation*}
	\begin{split}
		C_1 &= w \lor \lnot a\\
		C_2 &= x \lor \neg b
	\end{split}
	\quad\quad
	\begin{split}
    C_3 &= y \lor \neg w \lor \neg x\\
		C_4 &= h_1 \lor \neg c
	\end{split}
  \quad\quad
  \begin{split}
    C_5 &= z \lor \neg d \lor \neg h_1\\
    C_6 &= h_2 \lor \neg z
	\end{split}
  \quad\quad
	\begin{split}
		C_7 &= a \lor \neg c \lor \neg d\\
		C_8 &= \neg w \lor \neg y \lor \neg z \lor \neg h_2
	\end{split}
\end{equation*}
\begin{figure}[ht!]
  \def\xdist{1.3cm}
  \def\ydist{1.1cm}
  \begin{subfigure}{0.49\textwidth}
    \centering
    \begin{tikzpicture}[scale=0.9, transform shape]
      \useasboundingbox (-1.5,-0.5) rectangle (3.5*\xdist,3.5*\ydist);

      \node at (-1, 0*\ydist) {$\level = 1$};
      \node at (-1, 1*\ydist) {$\level = 2$};
      \node at (-1, 2*\ydist) {$\level = 3$};
      \node at (-1, 3*\ydist) {$\level = 4$};

      \node[vertex, decision, ck_a ]           (a)  at (0*\xdist, 0*\ydist) {$a$};
      \node[vertex, implied , ck_a, conflict]  (w)  at (1*\xdist, 0*\ydist) {$w$};
      \node[] (C1) at (1.35*\xdist, -0.35*\ydist) {\tiny $C_1$};

      \node[vertex, decision, ck_b ]           (b)  at (0*\xdist, 1*\ydist) {$b$};
      \node[vertex, implied , ck_b ]           (x)  at (1*\xdist, 1*\ydist) {$x$};
      \node[] (C2) at (1.35*\xdist, 0.65*\ydist) {\tiny $C_2$};
      \node[vertex, implied , ck_ab, conflict] (y)  at (2*\xdist, 1*\ydist) {$y$};
      \node[] (C3) at (2.35*\xdist, 0.65*\ydist) {\tiny $C_3$};

      \node[vertex, decision, ck_c ]           (c)  at (0*\xdist, 2*\ydist) {$c$};
      \node[vertex, implied , ck_c ]           (h1) at (1*\xdist, 2*\ydist) {$\bm{h_1}$};
      \node[] (C4) at (1.35*\xdist, 1.65*\ydist) {\tiny $C_4$};

      \node[vertex, decision, ck_d ]           (d)  at (0*\xdist, 3*\ydist) {$d$};
      \node[vertex, implied , ck_cd, conflict] (z)  at (1*\xdist, 3*\ydist) {$z$};
      \node[] (C6) at (2.35*\xdist, 2.65*\ydist) {\tiny $C_6$};
      \node[vertex, implied , ck_cd, conflict] (h2) at (2*\xdist, 3*\ydist) {$\bm{h_2}$};
      \node[] (C5) at (1.35*\xdist, 2.65*\ydist) {\tiny $C_5$};

      \draw (a) edge[myarr] (w) ;

      \draw (b) edge[myarr] (x);
      \draw (x) edge[myarr] (y);
      \draw (w)  edge[myarr, bend right=30] (y);

      \draw (c) edge[myarr] (h1);

      \draw (d) edge[myarr] (z) ;
      \draw (h1) edge[myarr] (z) ;
      \draw (z)  edge[myarr] (h2);

    \end{tikzpicture}
    \caption{{\small Conflict detected in clause $C_8$.}}
    \label{fig:motivation-conflict}
  \end{subfigure}
  \begin{subfigure}{0.5\textwidth}
    \centering
    \begin{tikzpicture}[scale=0.9, transform shape]
      \useasboundingbox (-1.5,-0.5) rectangle (3.5*\xdist,3.5*\ydist);

      \node at (-1, 1*\ydist) {$\level = 1$};
      \node at (-1, 2*\ydist) {$\level = 2$};
      \node at (-1, 3*\ydist) {$\level = 3$};

      \node[vertex, decision, ck_b ]           (b)  at (0*\xdist, 1*\ydist) {$b$};
      \node[vertex, implied , ck_b ]           (x)  at (1*\xdist, 1*\ydist) {$x$};

      \node[vertex, decision, ck_c ]           (c)  at (0*\xdist, 2*\ydist) {$c$};
      \node[vertex, implied , ck_c ]           (h1) at (1*\xdist, 2*\ydist) {$\bm{h_1}$};

      \node[vertex, decision, ck_d ]           (d)  at (0*\xdist, 3*\ydist) {$d$};
      \node[vertex, implied , ck_cd]           (z)  at (1*\xdist, 3*\ydist) {$z$};
      \node[vertex, implied , ck_cd]           (h2) at (2*\xdist, 3*\ydist) {$\bm{h_2}$};
      \node[vertex, implied , ck_bcd]          (nx) at (3*\xdist, 3*\ydist) {$\neg w$};
      \node at (3.35*\xdist, 2.65*\ydist) {\tiny $C_9'$};

      \draw (b) edge[myarr] (x);

      \draw (c) edge[myarr] (h1);

      \draw (d) edge[myarr] (z) ;

      \draw (h1) edge[myarr] (z) ;
      \draw (z)  edge[myarr] (h2);
      \draw (z)  edge[myarr, bend left=30] (nx);
      \draw (h2) edge[myarr] (nx);
      \draw (x) edge[myarr, bend right=30] (nx);

    \end{tikzpicture}
    \caption{{\small Graph backtracking (GB) [this work].}}
    \label{fig:motivation-gb}
  \end{subfigure}

  \begin{subfigure}{0.49\textwidth}
    \centering
    \begin{tikzpicture}[scale=0.9, transform shape]
      \useasboundingbox (-1.5,-0.5) rectangle (3.5*\xdist,2.5*\ydist);

      \node at (-1, 0*\ydist) {$\level = 1$};
      \node at (-1, 1*\ydist) {$\level = 2$};

      \node[vertex, decision, ck_a ]           (a)  at (0*\xdist, 0*\ydist) {$a$};
      \node[vertex, implied , ck_a ]           (w)  at (1*\xdist, 0*\ydist) {$w$};

      \node[vertex, decision, ck_b ]           (b)  at (0*\xdist, 1*\ydist) {$b$};
      \node[vertex, implied , ck_b ]           (x)  at (1*\xdist, 1*\ydist) {$x$};

      \node[vertex, implied , ck_ab]           (y)  at (2*\xdist, 1*\ydist) {$y$};
      \node[vertex, implied , ck_ab]           (nz) at (3*\xdist, 1*\ydist) {$\neg z$};
      \node at (3.35*\xdist, 0.65*\ydist) {\tiny $C_9$};

      \draw (a) edge[myarr] (w) ;

      \draw (b) edge[myarr] (x);
      \draw (x) edge[myarr] (y);
      \draw (w)  edge[myarr, bend right=30] (y);
      \draw (y) edge[myarr] (nz) ;
      \draw (w)  edge[myarr, bend right=25] (nz) ;

    \end{tikzpicture}
    \caption{{\small Non-chronological backtracking (NCB)~\cite{DBLP:conf/dac/MoskewiczMZZM01}.}}
    \label{fig:motivation-ncb}
  \end{subfigure}
  \begin{subfigure}{0.5\textwidth}
    \centering
    \begin{tikzpicture}[scale=0.9, transform shape]

      \useasboundingbox (-1.5,-0.5) rectangle (3.5*\xdist,2.5*\ydist);

      \node at (-1, 0*\ydist) {$\level = 1$};
      \node at (-1, 1*\ydist) {$\level = 2$};
      \node at (-1, 2*\ydist) {$\level = 3$};

      \node[vertex, decision, ck_a ]           (a)  at (0*\xdist, 0*\ydist) {$a$};
      \node[vertex, implied , ck_a ]           (w)  at (1*\xdist, 0*\ydist) {$w$};

      \node[vertex, decision, ck_b ]           (b)  at (0*\xdist, 1*\ydist) {$b$};
      \node[vertex, implied , ck_b ]           (x)  at (1*\xdist, 1*\ydist) {$x$};
      \node[vertex, implied , ck_ab]           (y)  at (2*\xdist, 1*\ydist) {$y$};
      \node[vertex, implied , ck_ab]           (nz) at (3*\xdist, 1*\ydist) {$\neg z$};
      \node at (3.35*\xdist, 0.65*\ydist) {\tiny $C_9$};

      \node[vertex, decision, ck_c ]           (c)  at (0*\xdist, 2*\ydist) {$c$};
      \node[vertex, implied , ck_c ]           (h1) at (1*\xdist, 2*\ydist) {$\bm{h_1}$};

      \draw (a) edge[myarr] (w) ;

      \draw (b) edge[myarr] (x);
      \draw (x) edge[myarr] (y);
      \draw (w)  edge[myarr, bend right=30] (y);
      \draw (y) edge[myarr] (nz) ;
      \draw (w)  edge[myarr, bend right=25] (nz) ;

      \draw (c) edge[myarr] (h1);

    \end{tikzpicture}
    \caption{{\small Chronological backtracking (CB)~\cite{DBLP:conf/sat/NadelR18}.}}
    \label{fig:motivation-cb}
  \end{subfigure}
	\caption{Implication graphs of $F$ for different backtracking strategies.
          Decisions are represented as squares and applied in lexicographic order. When relevant, implied literals are annotated with their justifications. Colors indicate dependencies on decision literals. Decision levels are annotated with $\level$. The literals of the conflicting clause $C_8$ are highlighted with thicker dashed borders.}
	\label{fig:motivation}
\end{figure}
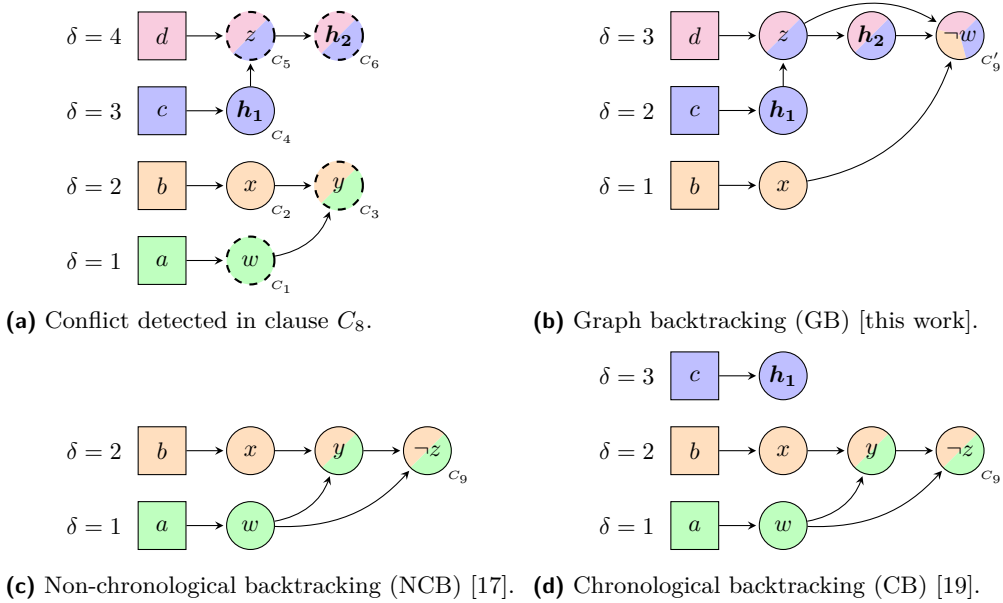

\Cref{fig:motivation-conflict} displays a possible implication graph resulting from $F$.
An implication graph denotes
logical implication relations between literals (for preliminaries, see \Cref{sec:preliminaries}).
A conflict is detected in clause $C_8$ during the propagation of $z$.
To repair it, the traditional conflict analysis~\cite{DBLP:journals/tc/Marques-SilvaS99}
algorithm learns the clause $C_9 = \neg w \lor \neg y \lor \neg z$ and non-chronological backtracking~\cite{DBLP:conf/dac/MoskewiczMZZM01} backjumps to the decision level of $b$ (\Cref{fig:motivation-ncb}). NCB notably unassigns the heavy literals $h_1$ and $h_2$ against the user's preference.
Chronological
backtracking~\cite{DBLP:conf/sat/NadelR18}
learns the same clause $C_9$ but realizes that the decision
$c$ does not need to be undone.
Instead, it would only undo $d, z$ and $h_2$ (\Cref{fig:motivation-cb}).
This is an improvement, as $h_1$ remains assigned.

Our graph backtracking algorithm allows unassigning \emph{any}%
\footnote{Some safeguards are needed to ensure termination (\Cref{sec:chunk-selection}).}
of the decisions involved in the conflict and their
subsequent implications, i.e., any of the colors in \Cref{fig:motivation-conflict}.
Considering the user-provided weights, it unassigns the green literals $a, w$ and $y$,
preserving both heavy literals $h_1$ and $h_2$, as illustrated in
\Cref{fig:motivation-gb}.
Our modified analysis (see \Cref{sec:chunk:repair}) learns%
\footnote{Shrinking could also be used to shorten the clause to $C_9' = \lnot w \lor \lnot x \lor \neg z$. We omit this optimization.}
the clause $C_9' = \lnot w \lor \lnot x \lor \neg z \lor \neg h_2$ and
GB now implies $\neg w$ because of $C_9'$ (\Cref{fig:motivation-gb}).


\paragraph{Contributions.}
Our main contributions are as follows:
\begin{itemize}
	\item We extend the interface of SAT solvers with a weight function for literals (\Cref{sec:chunk-selection}).
	\item We present \emph{graph backtracking}, a precise and user-guided backtracking algorithm (\Cref{alg:cdcl}).
	\item We show that graph
          backtracking is sound and terminating (\Cref{thm:gb:soundness}).
  \item We show that graph backtracking is a generalization of (N)CB (\Cref{lem:cb-simulation,lem:ncb-simulation}).
	\item We suggest several practical optimizations of graph backtracking (\Cref{sec:gb:improve}).
	\item We implement
          graph backtracking and evaluate its performance empirically (\Cref{sec:experiments}).
\end{itemize}


\section{Preliminaries}
\label{sec:preliminaries}

We assume familiarity with propositional logic and CDCL~\cite{DBLP:series/faia/336} and use the standard logical connectives $\lnot$, $\land, \lor$, and $\Rightarrow$.
We denote a finite set of propositional variables by $\mathcal{V}$. A \emph{literal} is either a variable or its negation; that is, a literal is  $v$ or $\lnot v$, for some $v\in \mathcal{V}$.
The set of literals is denoted by $\mathcal{L} = \{v, \lnot v:v \in \mathcal{V}\}$. Note that $\lnot\lnot \ell = \ell$.
A \emph{clause} is a (disjunctive) set of literals $C = \{c_1,c_2,\dots,c_n\}$ where $c_i$ is a literal.
We denote the empty clause as $\square$ and the undefined clause $\blacksquare$ is used for the absence of a clause in partial functions.
Without loss of generality, we only consider a propositional formula $F$ in \emph{conjunctive normal form (CNF)} as a (conjunctive) set of clauses $\{C_1, C_2,\dots,C_m\}$ over $\mathcal{V}$.

An \emph{ordered set} is a set with a total positional order on its elements.
Ordered sets support the concatenation operator $\cdot$ and are stable under set operations, i.e., set intersection and difference on ordered sets maintain the order of the left operand.

A (partial) assignment $\p$ is an ordered set of literals.
A literal and its negation cannot appear simultaneously in $\p$.
In an assignment $\p$, a literal is \emph{satisfied} if $\ell \in \p$; \emph{falsified} if $\neg \ell \in \p$; and \emph{unassigned} otherwise.
We say that a clause is \emph{satisfied} if $C \cap \p \neq \emptyset$; \emph{falsified} if $C \setminus \{\neg \ell:\ell \in \p\} = \square$; and \emph{unit} in literal $\ell$ if $C \setminus \{\neg \ell':\ell' \in \p\} = \{\ell\}$.

A literal $\ell$ is \emph{implied} by a clause $C$ under a partial assignment $\p$ if $C$ is unit in $\ell$ under $\p$ and $\ell$ is satisfied by $\p$.
The \emph{justification} of $\ell$, also known as the \emph{reason} of $\ell$, is the unit clause $\reason(\ell) = C$ that explains the implication of $\ell$ on $\p$.
A \emph{missed implication} is a clause $C$ that implies a literal $\ell$ under $\p$ but is not the justification of $\ell$.
A \emph{missed lower implication} is a missed implication $C$ of $\ell$ such that $\ell$ is implied by $C$ at a lower decision level than $\reason(\ell)$.
A \emph{decision} literal $\ell$ does not have a justification, hence $\reason(\ell) = \blacksquare$.

The \emph{decision level} $\level(\ell)$ of a literal $\ell$ is defined as follows: if $\ell$ is a decision, then its decision level is the number of decisions in $\p$ preceding $\ell$, plus one; otherwise, its decision level is the maximum level of all literals in $\reason(\ell) \setminus \{\ell\}$.
When relevant, we write $\ell @ n$ to denote that a literal $\ell$ is on decision level $n$.
The decision level of a set of literals $S$ is $\level(S) = \max_{\ell \in S} \level(\ell)$.
Furthermore, we define a \emph{literal weight function}
$\weight : \mathcal{L} \rightarrow \mathbb{R}$ that maps a literal $\ell\in\mathcal{L}$ to a real number. We also extend $\weight$ to sets of literals by defining $\weight(S) = \sum_{\ell \in S} \weight(\ell)$.

Given a partial assignment $\p$, we partition $\p = \trail \cdot \q$ into the set of propagated literals $\trail$ and the queue of literals to propagate $\q$.
We say that a literal $\ell$ is propagated if $\ell \in \trail$ or $\neg \ell \in \trail$. We note that some related work calls \emph{propagated} what we call \emph{implied}; this distinction becomes relevant in \Cref{sec:BCP}.

To visualize the implications of literals on a partial assignment $\p$, we use \emph{implication graphs}, as shown in \Cref{fig:motivation}. An implication graph is a directed acyclic graph
with a node for each literal in $\p$ and edges indicating the justification of each literal.
A literal $\ell$ in the graph is the implied consequence of all literals with an edge to $\ell$ using clause $\reason(\ell)$.
Decision literals do not have any incoming edges.
A literal \emph{$\ell$ depends on literal $\ell'$} if there is a path from $\ell'$ to $\ell$ in the implication graph.


\section{Related Work}
\label{sec:related}
\paragraph{Non-Chronological Backtracking (NCB)} has been the standard backtracking approach in most state-of-the-art SAT solvers~\cite{DBLP:series/faia/336,DBLP:conf/sat/EenS03} since watched literals became popular~\cite{DBLP:conf/dac/MoskewiczMZZM01}.
When a NCB algorithm finds a conflict, it performs conflict analysis to learn a new clause preventing the same conflict from happening again. The standard algorithm searches a unique implication point (UIP), i.e., a clause with one literal at the highest decision level. The UIP results from binary resolutions of the conflicting clause and the justifications leading to it. Usually, we stop at the first UIP encountered (1UIP).

NCB is characterized by backjumping. To repair the conflict, NCB will backtrack to the second-highest level in the learned clause. In the example \Cref{fig:motivation-conflict}, the clause $C_9 = \neg w@1 \lor \neg y@2 \lor \neg z@4$ has the highest decision level $4$ and second-highest decision level $2$. The algorithm then backjumps to decision level $2$ (\Cref{fig:motivation-ncb}).

\paragraph{Chronological Backtracking (CB)} was reintroduced more recently~\cite{DBLP:conf/sat/NadelR18,DBLP:conf/sat/MohleB19,DBLP:conf/sat/Nadel22,DBLP:conf/sat/CoutelierFK24}. It won the SAT2018 competition with \textsc{Maple\_LCM} and is now part of \cadical~\cite{DBLP:conf/cav/BiereFFFFP24}. It was shown that CB also improves performance in applications such as model counting~\cite{DBLP:conf/cav/SoosM25} and enumeration~\cite{DBLP:journals/ai/SpallittaSB25} due to a more systematic exploration of the search space, but also MaxSAT~\cite{DBLP:conf/sat/Nadel22} due to its higher flexibility in incremental solving.

Less restrictive than NCB, CB only requires backtracking to the highest decision level in the learned clause, minus one. In the example \Cref{fig:motivation-conflict}, the clause $C_9$ has the highest decision level $4$. The algorithm then backtracks to decision level $3$ (\Cref{fig:motivation-cb}).

More precisely, CB may backtrack anywhere between the second highest and the highest decision level minus one, and heuristics can be used to select the appropriate level \cite{DBLP:conf/sat/NadelR18,DBLP:conf/sat/MohleB19}. In the following, we will only consider the variant of CB that backtracks to the highest decision level minus one, as it is simpler to differentiate from NCB.

\paragraph{Beyond SAT Solving.}
Some attempts have been made to improve the backtracking processes in CSP, notably~\cite{DBLP:journals/jair/Ginsberg93}. However, these approaches did not consider the inner workings of propagation algorithms and did not learn clauses. Additionally, \cadical introduced a \emph{facade} for better integration in SMT solver with CB~\cite{DBLP:conf/cav/BiereFFFFP24}. \cadical shows a partial assignment to the SMT solver that does not reflect the state of the SAT solver and pretends to use a stack. This facade is orthogonal to our approach, and could be used in conjunction with graph backtracking. Some work has been done to add implications graphs with redundancy in the SMT congruence closure algorithm~\cite{DBLP:conf/vmcai/AndreottiB26}.
GB might also have applications in local search~\cite{DBLP:conf/aaai/CaiS12} due to its minimalistic approach to backtracking.


\section{Graph Backtracking}
\label{sec:details}
We now present our \emph{graph backtracking (GB)} algorithm. We highlight the
 key features of CDCL-based SAT solving using our GB approach,
 in particular when compared to the existing strategies of chronological backtracking (CB)~\cite{DBLP:conf/sat/NadelR18} and non-chronological backtracking (NCB)~\cite{DBLP:journals/tc/Marques-SilvaS99}.
 To ease readability, the differences between GB and (N)CB are highlighted in blue in the respective algorithms.

 \paragraph{\bf GB in a Nutshell.} The main idea of GB is
 to use implication graphs to determine the set of literals that can
 be unassigned upon conflict detection. We cluster literals
 into so-called \emph{chunks}: groups of literals that can be undone together.
 We can repair the assignment by undoing any chunk involved in the conflict (\Cref{sec:chunk-selection}).
 With this strategy, we obtain a CDCL-based SAT solving approach using GB,
 summarized in \Cref{alg:cdcl},
 where treatment of propositional variables in case of
 conflict analysis is more surgical and controlled.

In (N)CB, literals are clustered according to their decision levels. This clustering forms a coarse over-approximation of the actual dependencies between literals. The decision levels create the simplifying assumption that literals assigned at a level $d$ depend on all literals below $d$. Hence, when backtracking to level $d$, all literals above $d$ need to be unassigned.
Our GB algorithm refines this clustering by using \emph{chunks}. A chunk $ck_\ell$ is a set of literals depending on the decision literal $\ell$.
That is, the set of literals reachable from a decision $\ell$ in the implication graph. Naturally, an implied literal can depend on multiple decisions and belong to several chunks. We write $\chunks(\ell)$ as the set of chunks containing a literal $\ell$. For simplicity, we assume that $\chunks(\ell) = \chunks(\neg \ell)$. Our notation extends to sets of literals (such as clauses) $S$ as $\chunks(S) = \bigcup_{\ell \in S} \chunks(\ell)$. A decision literal $d$ only belongs to its associated chunk $\chunks(d) = \{ck_d\}$. A literal $\ell$ justified by a clause $C = \reason(\ell)$ belongs to the chunks $\chunks(\ell) = \chunks(C\setminus \{\ell\})$.

\begin{example}
  \label{ex:gb:chunk-selection}
  Consider \Cref{fig:motivation-conflict}. The implication graph contains four chunks, each highlighted in a distinct color, $\green{ck_a} = \{a, w, y\}, \orange{ck_b} = \{b, x, y\}, \blue{ck_c} = \{c, z, h_1, h_2\}$, and $\magenta{ck_d} = \{d, z, h_2\}$. The conflicting clause $C_8 = \neg w \lor \neg y \lor \neg z \lor \neg h_2$ belongs to the four chunks $\chunks(C_8) = \{ck_a, ck_b, ck_c, ck_d\}$. GB can therefore choose to backtrack $ck_a$ since it is the lightest, as it does not contain any heavy literals $h_i$. (\Cref{sec:chunk-selection} explains why $ck_b$ is not selected.)
\end{example}

When a conflict clause $C$ is discovered, in GB we may undo any of the chunks in $\chunks(C)$ to resolve the conflict. \Cref{sec:chunk-selection} describes how this choice is made. The undone chunks influence the conflict analysis procedure of GB, as explained in \Cref{sec:chunk:repair}. To maintain clean invariants, watched literals require special care in GB, as detailed in \Cref{sec:BCP}.


\subsection{Repair Chunk Selection}
\label{sec:chunk-selection}

Upon discovering a conflict $C$, a CDCL solver needs to undo parts of the assignment, such that a newly learned clause $C'$, derived from the conflict, becomes unit and implies a flipped literal. To this end, GB may undo any set of chunks%
\footnote{Here we talk about chunk sets and not individual chunks for generality, relevant in \Cref{sec:gb:improve,sec:reduction}.}
$\Gamma_i$ that contains exactly one%
\footnote{Undoing more chunks may make it impossible to learn a unit clause.}
chunk from $\chunks(C)$, i.e., $|\Gamma_i \cap \chunks(C)| = 1$.
GB then selects the lightest chunk set $\Gamma^*$  among all valid candidates $\Gamma_i$. Conflict analysis then discovers $C'$ containing exactly one literal $\ell$ belonging to $\Gamma^*$. After backtracking $\Gamma^*$, $\ell$ is implied by $C'$.%

Among all backtrack candidates, the selection of $\Gamma^*$ is guided by the 
user-provided cost function $\weight$.
A careless choice of backtracked chunks might not lead to termination, as illustrated next.

\begin{figure}[h]
  \centering
  \def\xdist{1.2cm}
  \def\ydist{1.2cm}
  \vspace{1em} 
  \begin{subfigure}{0.30\textwidth}
    \centering
    \begin{tikzpicture}[scale=0.9, transform shape]
      \useasboundingbox (-1.5,-1) rectangle (3,2.4);
      \node at (-1.1, 0) {$\level = 1$};
      \node at (-1.1, 1.8) {$\level = 2$};

      \node[vertex, decision, ck_a          ](a) at (0*\xdist, 0*\ydist)   {$a$};
      \node[vertex, implied,  ck_a          ](w) at (1*\xdist, 0*\ydist)   {$w$};
      \node[vertex, implied,  ck_a, conflict](z) at (2*\xdist, 0*\ydist)   {$z$};
      \node at (2.35*\xdist, -0.35*\ydist) {\tiny $D_2$};
      \node[vertex, decision, ck_b          ](b) at (0*\xdist, 1.5*\ydist) {$b$};
      \node[vertex, implied,  ck_b, conflict](x) at (1*\xdist, 1*\ydist)   {$x$};
      \node[vertex, implied,  ck_b, conflict](y) at (1*\xdist, 2*\ydist)   {$y$};

      \draw (b) edge[myarr] (x);
      \draw (b) edge[myarr] (y);
      \draw (a) edge[myarr] (w);
      \draw (a) edge[myarr, bend left] (z);
      \draw (w) edge[myarr] (z);
    \end{tikzpicture}
    \caption{Conf. $D_1 = \neg x \lor \neg y \lor \neg z$}
    \label{fig:proof-conflict}
  \end{subfigure}
  \hfill
  \begin{subfigure}{0.30\textwidth}
    \centering
    \begin{tikzpicture}[scale=0.9, transform shape]
      \useasboundingbox (-1.5,-1) rectangle (3,2.4);

      \node at (-1.1, 1.8) {$\level = 1$};

      \node[vertex, decision, white](a) at (0*\xdist, 0*\ydist)   {$a$};
      \node[vertex, implied,  white](w) at (1*\xdist, 0*\ydist)   {$w$};
      \node[vertex, decision, ck_b](b)  at (0*\xdist, 1.5*\ydist) {$b$};
      \node[vertex, implied,  ck_b](x)  at (1*\xdist, 1*\ydist)   {$x$};
      \node[vertex, implied,  ck_b](y)  at (1*\xdist, 2*\ydist)   {$y$};
      \node[vertex, implied,  ck_b](nz) at (2*\xdist, 1.5*\ydist) {$\neg z$};
      \node at (2.35*\xdist, 1.15*\ydist) {\tiny $D_1$};

      \draw (b) edge[myarr] (x);
      \draw (b) edge[myarr] (y);
      \draw (x) edge[myarr] (nz);
      \draw (y) edge[myarr] (nz);
    \end{tikzpicture}
    \caption{Undo chunk $ck_a$.}
    \label{fig:proof-backtrack}
  \end{subfigure}
  \hfill
  \begin{subfigure}{0.30\textwidth}
    \centering
    \begin{tikzpicture}[scale=0.9, transform shape]
      \useasboundingbox (-1.5,-1) rectangle (3,2.4);
      \node at (-1.1, 0) {$\level = 1$};
      \node at (-1.1, 1.8) {$\level = 2$};

      \node[vertex, decision, ck_c, conflict](a)  at (0*\xdist,  1.5*\ydist) {$a$};
      \node[vertex, implied , ck_c, conflict](w)  at (1*\xdist,  1.5*\ydist) {$w$};
      \node[vertex, decision, ck_b](b)            at (0*\xdist,    0*\ydist) {$b$};
      \node[vertex, implied,  ck_b](y)            at (1*\xdist,  0.5*\ydist) {$y$};
      \node[vertex, implied,  ck_b](x)            at (1*\xdist, -0.5*\ydist) {$x$};
      \node[vertex, implied,  ck_b, conflict](nz) at (2*\xdist,    0*\ydist) {$\neg z$};
      \node at (2.35*\xdist, -0.35*\ydist) {\tiny $D_1$};

      \draw (b) edge[myarr] (x);
      \draw (b) edge[myarr] (y);
      \draw (a) edge[myarr] (w);
      \draw (x) edge[myarr] (nz);
      \draw (y) edge[myarr] (nz);
    \end{tikzpicture}
    \caption{Conf. $D_2 = \neg a \lor \neg w \lor z$.}
    \label{fig:proof-loop}
  \end{subfigure}
  \caption{Potential loop in conflict discovery with arbitrary backtracking.}
  \label{fig:proof-loop-example}
\end{figure}
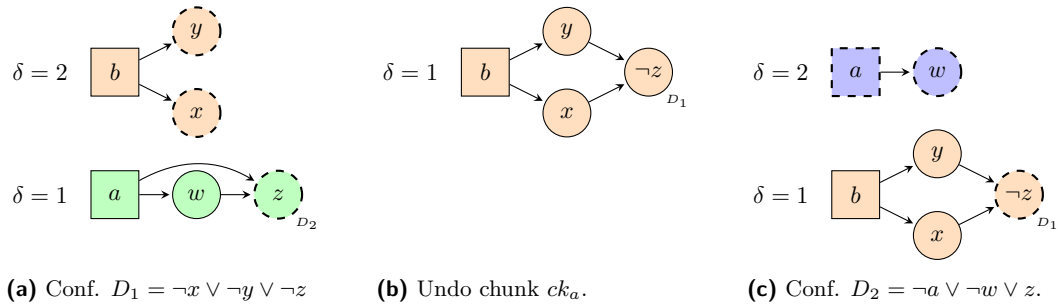

\begin{example}\label{ex:GB:cost:safeguards}
  Consider \Cref{fig:proof-loop-example}. The implication graph contains two decision variables $a$ and $b$, with their respective chunks $ck_a$ and $ck_b$.
  On \Cref{fig:proof-conflict}, BCP (\Cref{sec:BCP}) discovers a conflict clause $D_1 = \neg x \lor \neg y \lor \neg z$, which has a single literal in the chunk $ck_a$. The chunk $ck_a$ can then be undone as shown in \Cref{fig:proof-backtrack} to fix the conflict. The literal $\neg z$ is added to $ck_b$. CDCL might decide $a$ again, and discovers another conflict $D_2 = \neg a \lor \neg w \lor z$ (\Cref{fig:proof-loop}). If now $ck_b$ is undone, and the decision $b$ is taken, we are back to the same state as \Cref{fig:proof-conflict}.
\end{example}

As further elaborated in \Cref{sec:chunk:repair}, the choice of chunk for backtracking influences the clause that is learned after conflict analysis.
We write 1UIP$_{\Gamma_i}(C)$ to denote the first UIP found when unassigning chunks $\Gamma_i$ for the conflict clause $C$.
The \emph{chunk level} of a chunk $ck_\ell$, written $\level_{\chunks}(ck_\ell)$, is the decision level $\level(\ell)$ of its decision literal $\ell$.
We extend the notation to sets of chunks $\Gamma_i$ as $\level_{\chunks}(\Gamma_i) = \max_{ck \in \Gamma_i} \level_{\chunks}(ck)$.

\paragraph{Safeguard.}
The loop in \Cref{ex:GB:cost:safeguards} arises because neither conflict leads to learning a new clause. In \Cref{fig:proof-conflict}, $D_1$ is already a UIP with respect to $ck_a$, i.e., 1UIP$_{ck_a}(D_1) = D_1$. Later, in \Cref{fig:proof-loop}, $D_2$ is also a UIP with respect to $ck_b$.

Redundant clauses can also be learned in CB. However, progress is ensured by the ordering of decision levels.
Indeed, in CB, a literal flipped due to a conflict will remain on the trail at least until another flipped literal is added at a lower level. It is then impossible to encounter infinitely many conflicts before reaching decision level 0, and progress is guaranteed. 

In GB, this ordering is weakened as follows: candidate chunks either lead to learning a new clause, or contain the most recent chunk in the conflict.

\begin{definition}[Terminating backtrack candidates\label{def:term:bCands}]
  Let $C$ be a conflicting clause in the clause set $F$.
Let $\chunks(\p)$ be the set of all chunks in the current partial assignment $\p$, and $2^{\chunks(\p)}$ be the powerset of $\chunks(\p)$. The set $\Gamma$ of \emph{terminating backtrack candidates} is defined as:
\begin{equation*}
  \Gamma = \left\{\Gamma_i \in 2^{\chunks(\p)}: |\Gamma_i \cap \chunks(C)| = 1 \land \Bigl(\text{\normalfont{1UIP}}_{\Gamma_i}(C) \notin F \lor \level_{\chunks}\bigl(\Gamma_i \cap \chunks(C)\bigr) = \level(C)\Bigr)\right\}
\end{equation*}
\end{definition}

Selecting $\Gamma^*$ from the terminating backtrack candidates guarantees the termination of GB while allowing the preservation of heavy literals.
The first condition $\text{\normalfont{1UIP}}_{\Gamma_i}(C) \notin F$ ensures that the learned clause is new. Alternatively, if the second condition $\level_{\chunks}(\Gamma_i \cap \chunks(C)) = \level(C)$ holds, the chunks are ordered and an identical argument to CB applies.

\begin{remark}
  It is not always trivial to check the redundancy of the learned clause. Let $C_1 \in F$ be the conflict detected by \Call{BCP}{} and $C_2 \in F$ be another conflict. If $\text{\normalfont{1UIP}}_{\Gamma_i}(C_1) = C_2$, then the first condition fails. A simple way to detect such a situation is to go through the watch lists of the literals in the learned clause and check whether it is subsumed in the formula.
  However, this case is very rare in practice and can be ignored in an efficient implementation.
\end{remark}

In \Cref{ex:GB:cost:safeguards}, backtracking $ck_a$ would not learn a new clause and $ck_b$ is more recent in the conflict $\level_{\chunks}(ck_a) < \level(C)$. Therefore $ck_a$ is not a terminating backtrack candidate.

In \Cref{ex:gb:chunk-selection}, our safeguard explains why we did not select $ck_b$ for backtracking. $C_8$ is already a UIP with respect to $ck_b$. Backtracking $ck_b$ does not learn a new clause and $ck_b$ is not the most recent chunk in the conflict.

\paragraph{Chunk Selection.} Selecting the chunk set $\Gamma^*$ becomes a minimization problem over the set of terminating backtrack candidates $\Gamma$. That is:
\begin{equation}
  \label{eq:chunk-selection}
  \Gamma^* = \Call{SelectChunk}{C, \weight} = \argmin_{\Gamma_i \in \Gamma} \weight\left(\bigcup_{ck \in \Gamma_i} ck\right)
\end{equation}

\begin{remark}
  Note that, if the cost function $\weight$ only returns positive values, \Call{SelectChunk}{} only selects backtrack candidates with a single chunk. However, in some cases, it might be beneficial to backtrack multiple chunks together (see \Cref{sec:reduction,sec:chunk-merging}).
\end{remark}


\subsection{Conflict Analysis \& Backtracking}\label{sec:chunk:repair}
Our conflict analysis procedure within GB is summarized in
\Cref{alg:conflict-analysis}, showcasing a variant of the standard 1UIP scheme~\cite{DBLP:journals/tc/Marques-SilvaS99} adapted to chunks. Binary resolution between the clauses $C$ and $D$ over a literal $\ell$ is denoted by $C \otimes_\ell D$, which is defined as $(C \setminus \{\neg \ell\}) \cup (D \setminus \{\ell\})$.

\begin{example}\label{ex:gb:analyze}
  In \Cref{fig:motivation-conflict}, the call \Call{Analyze}{$C_8, \{ck_a\}$} in
  \Cref{alg:conflict-analysis}
  handles $y$ first, since it is the last literal on the trail belonging to $ck_a$. It will perform resolution on $\reason(y) = C_3$, giving the clause $C_9' = C_8 \otimes_y C_3 = \neg w \lor \neg x \lor \neg z \lor \neg h_2$. The procedure ends with the learned clause $C_9'$ since $w$ is the last literal belonging to $ck_a$.
\end{example}


\begin{algorithm}[t]
  \caption{Conflict Analysis.}
  \label{alg:conflict-analysis}
  \begin{algorithmic}[1]
  \Procedure{Analyze}{$C, \blue{\Gamma^*}$}
    \State $D \gets C$
    \Comment{\parbox{0.5\textwidth}{Current learned clause.}}
    \State $n \gets \left|\blue{D \cap \bigcup_{ck \in \Gamma^*} ck}\right|$
    \Comment{\parbox{0.5\textwidth}{Number of literals in $D$ that will be undone}}
    \For{$\ell \in \p$}
      \Comment{\parbox{0.5\textwidth}{Iterate right to left}}
      \If {$\ell \in D \land \blue{\Gamma^* \cap \chunks(\ell) \neq \emptyset}$}
        \State $D \gets D\otimes_\ell \reason(\ell)$
        \Comment{\parbox{0.5\textwidth}{Binary Resolution}}
        \State $n \gets \left|\blue{D \cap \bigcup_{ck \in \Gamma^*} ck}\right|$
      \EndIf
      \If {$n = 1$}
        \Comment{\parbox{0.5\textwidth}{Found a UIP}}
        \State \Return $D$
      \EndIf
    \EndFor
  \EndProcedure
\end{algorithmic}

\end{algorithm}

As shown in \Cref{ex:gb:analyze}, chunk-based conflict analysis in GB adapts to any chunk and produces an implying learned clause.
The backtracking procedure of GB is presented in \Cref{alg:backtracking}.
We can now retain literals whose decision levels are higher than the backtracked chunk. In \Cref{fig:motivation-gb}, we keep $b, c, d, y, z, x, h_1$ and $h_2$ in the trail, whose decision levels are higher than $a$. This is why at line~\ref{line:backtrack:update-level} of \Cref{alg:backtracking}, we need to recalculate the decision level of each literal after the lowest undone decision.
Line~\ref{line:backtrack:repropagate} of \Cref{alg:backtracking} ensures the correctness of the two-watched-literal scheme in GB (\Cref{thm:gb:soundness}), as explained next.

\begin{algorithm}[t]
  \caption{Backtracking.}
  \label{alg:backtracking}
  \begin{algorithmic}[1]
  \Procedure{Undo}{\blue{$\Gamma^*$}}
    \For{$\ell \in \p$}
    \Comment{\parbox{0.45\textwidth}{Iterate left to right}}
        \If{\blue{$\chunks(\ell) \cap \Gamma^* \neq \emptyset$}}
            \State $\p \gets \p \setminus \{\ell\}$
            \Comment{\parbox{0.45\textwidth}{Unassign $\ell$}}
            \State $\chunks(\ell) \gets \emptyset,$
                   $\reason(\ell) \gets \blacksquare,$
                   $\level(\ell) \gets \infty$
            \State \Continue
        \EndIf
        \State \blue{$\level(\ell) \gets \Call{UpdateLevel}{\ell}$}
        \label{line:backtrack:update-level}
        \Comment{\parbox{0.45\textwidth}{Recalculate the decision level of $\ell$}}
        \If{\blue{$\crossChunks(\ell) \cap \Gamma^* \neq \emptyset$}}
            \Comment{\parbox{0.45\textwidth}{see \Cref{sec:BCP}}}
            \State $\blue{\trail \gets \trail \setminus \{\ell\}}$
            \label{line:backtrack:repropagate}
            \Comment{\parbox{0.45\textwidth}{$\ell$ needs to be repropagated}}
        \EndIf
    \EndFor
    \State $\trail \gets \p \cap \trail$
    \Comment{\parbox{0.45\textwidth}{Remove the unassigned literals}}
    \State $\q \gets \p \setminus \trail$
  \EndProcedure
\end{algorithmic}

\end{algorithm}

\subsection{Unit Propagation \& Watch Literals}
\label{sec:BCP}
One of the reasons for the success of NCB in CDCL-based SAT solvers is the use of fast unit propagation algorithms.
The most common technique to achieve Boolean Constraint Propagation (BCP) is the two-watched-literal scheme~\cite{DBLP:conf/dac/MoskewiczMZZM01}.
Each clause is said to be watched by two of its literals -- called \emph{watched literals} or \emph{watchers} -- such that the solver only needs to look at clauses watched by the negation of the currently being propagated literal to ensure the discovery of every conflict and unit clause.
We write $\wl(\ell)$ to denote the watch list of a literal $\ell$, that is the list of clauses watched by $\ell$. The following invariant is the guiding principle for existing backtracking solutions:

\begin{invariant}[Watched literals~\cite{DBLP:conf/sat/CoutelierFK24}]
    \label{inv:watched-literals}
    Consider the trail \(\p = \trail \cdot \q\). For each clause $C \in F$ watched by $c_1, c_2$, we have
    $\neg c_1 \in \trail \Rightarrow c_2 \in \p$.
\end{invariant}
It states that for each clause $C$, if one of its watchers $c_1$ is falsified and propagated, then the other watcher $c_2$ must be satisfied. During propagation, \Call{BCP}{} checks whether a literal $\ell$ can be added to the propagated set, i.e., whether it is possible to enforce $\neg c_1 \in (\trail \cdot \ell) \Rightarrow c_2 \in \p$. If this is not possible, then a conflict is raised.
Note that conflicts do not violate the invariant, since they are detected while the watched literals are still in the queue $\q$.
\Cref{inv:watched-literals} is not trivially preserved in GB upon backtracking.
We, thus, revise it into \cref{inv:backtrack-compatible-watched-literals} in order to show soundness of GB in \Cref{thm:gb:soundness}.

\paragraph{\bf Limitations of Watched Literals in GB.}
Invariant~\ref{inv:watched-literals} is naturally maintained by NCB, since the partial assignment behaves exactly like a stack.
If during the propagation of a literal $\ell$ a clause is detected to be satisfied, it will remain so at least until $\ell$ is unassigned. That is, the right side $c_2 \in \p$ holds at least as long as the left side $\neg c_1 \in \trail$ holds. Therefore, backtracking does not violate the invariant.

This property does not apply to GB by construction. In \Cref{fig:motivation}, we propagated $x$ before $y$, yet $y$ was unassigned first.
Therefore, backtracking does not preserve \Cref{inv:watched-literals}. The clause $C_3 = y \lor \neg x \lor \neg w$ would satisfy \Cref{inv:watched-literals} if $c_1 = \neg x$ and $c_2 = y$, but no longer after backtracking since $y \notin \p$.
To restore it, $x$ needs to be removed from the propagated set $\trail$.

A surprisingly effective, yet simple, solution used in CB is to repropagate all literals in the trail located after the first literal of the backtracked level, called the restoration point~\cite{DBLP:conf/sat/NadelR18,DBLP:conf/sat/CoutelierFK24}.
In CB, repropagation is cheap because watch lists are already cleaned up, and the restoration point is usually relatively close to the end of the trail.
However, in GB, this point is not trivial to find, and may be located early, leading to a lot of unnecessary propagations.

Another interesting aspect of GB is that there is no total ordering between literals. In (N)CB, it is possible to choose literals with the highest decision level to ensure that the watchers of a clause will always be unassigned before the other literals in the clause. In GB, this is no longer the case. Consider again $C_3$ in the example of \Cref{fig:motivation-conflict}, we would like to select the watchers such that they will be unassigned after backtracking. It is clear that $y$ should be one of the watched literals, since it belongs to both $ck_a$ and $ck_b$. However, we can either choose $\neg w$ or $\neg x$ as the second watcher. Let us select $\neg w$. In this case, after backtracking $ck_b$, the clause violates \Cref{inv:watched-literals}, since $\neg w \in \trail$ but $y \notin \p$. If we had selected $\neg x$ instead, backtracking $ck_a$ would also violate the invariant. We cannot find a perfect watcher since $\chunks(w) \nsubseteq \chunks(x)$ and $\chunks(x) \nsubseteq \chunks(w)$.

\paragraph{\bf Cross-Chunks of Watched Literals in GB.}
To solve the aforementioned issues on watch literals in GB, we introduce the concept of \emph{cross-chunks}, written as $\crossChunks(\ell)$ and used in \Cref{alg:backtracking}.
The cross-chunks of a literal $\ell$ is the set of chunks that, when backtracked, require $\ell$ to be repropagated.
In other words, if $ck \in \crossChunks(\ell)$ and $ck$ is undone, then $\ell$ needs to be repropagated; as ensured in line~\ref{line:backtrack:repropagate} of \Cref{alg:backtracking}:
The literal $\ell$ is removed from the propagated set $\trail$ and put
in the queue $\q$ to be repropagated.
Cross-chunks allow skipping most of the repropagations.
The backtracking procedure of \Cref{alg:backtracking}
preserves the following stronger invariant:
\newpage
\begin{invariant}[GB watched literals]
  Consider the trail \(\p = \trail \cdot \q\). For each clause $C \in F$ watched by $c_1, c_2$, we have
    $\neg c_1 \in \trail \Rightarrow \left[c_2 \in \p \land \chunks(c_2) \subseteq \bigl(\chunks(c_1) \cup \crossChunks(c_1)\bigr)\right]$.
  \label{inv:backtrack-compatible-watched-literals}
\end{invariant}

The revised invariant states that if a watched literal $c_1$ is falsified and propagated, then the other watcher $c_2$ must be satisfied and $c_2$ will hold as long as $\neg c_1$ is in the propagated set. That is, if $c_2$ is unassigned during backtracking, then $c_1$ will either be unassigned as well, or repropagated. During \Call{BCP}{}, we update the cross-chunks of the watched literals to keep a record of literals that need to be repropagated after backtracking.

\begin{algorithm}[t]
  \caption{Boolean Constraint Propagation.}
  \label{alg:bcp}
  \begin{algorithmic}[1]
  \Procedure{BCP}{ }
    \While {$\q \neq \emptyset$}
      \State $\ell \gets \q[0],$
             $c_1 \gets \neg \ell$
      \State \blue{$\crossChunks(\ell) \gets \chunks(\ell)$}
      \label{line:prop:reset-eta}
      \Comment{\parbox{0.45\textwidth}{Reset and enforce $\chunks(\ell) \subseteq \crossChunks(\ell)$}}
      \For{$C \in \wl(c_1)$}
        \State $c_2 \gets $ the other watched literal in $C$
        \If {$c_2 \in \p \blue{\land \chunks(c_2) \subseteq \crossChunks(c_1)}$}
          \label{line:prop:skip-condition}
          \State \Continue
          \label{line:prop:watched-literal-satisfied}
          \Comment{\parbox{0.45\textwidth}{C already satisfies \Cref{inv:backtrack-compatible-watched-literals}}}
        \EndIf

        \State $r \gets$ \Call{SearchReplacement}{$C, c_1, c_2$}
        \State $\wl(c_1) \gets \wl(c_1) \setminus \{C\}$
        \Comment{\parbox{0.45\textwidth}{Evict $C$ from the watch list of $c_1$}}
        \State $\wl(r) \gets \wl(r) \cup \{C\}$
        \Comment{\parbox{0.45\textwidth}{Add $C$ to the watch list of $r$}}
        \If {$\neg r \notin \p$}
          \Comment{\parbox{0.45\textwidth}{Good replacement found}}
          \State \Continue
          \label{line:prop:watched-literal-replaced}
        \EndIf
        \label{line:prop:watched-literal-replaced-higher-level}
        \If {$\neg c_2 \in \p$}
        \Comment{\parbox{0.45\textwidth}{Conflict}}
          \State \Return $C$
          \label{line:prop:conflict}
        \EndIf
        \If{\blue{$c_2 \in \p$}}
          \label{line:prop:missed-cross-chunk-implication-start}
          \Comment{\parbox{0.45\textwidth}{Missed implication}}
          \State \blue{ $\crossChunks(r) \gets \crossChunks(r) \cup \chunks(c_2)$}
          \State \blue{\Continue}
        \EndIf
        \label{line:prop:missed-cross-chunk-implication-end}
        \State $\q \gets \q \cdot c_2, \reason(c_2) \gets C, \level(c_2) \gets \level(C \setminus \{c_2\})$
        \State \blue{
          $\chunks(c_2) \gets \chunks(C \setminus \{c_2\})$,
          $\crossChunks(r) \gets \crossChunks(r) \cup \chunks(c_2)$}
        \label{line:prop:cross-chunk-implication}
      \EndFor
      \State $\q \gets \q \setminus \{\ell\}, \trail \gets \trail \cdot \ell$
    \EndWhile
  \State \Return $\blacksquare$
  \EndProcedure
\end{algorithmic}

\end{algorithm}

\paragraph{\bf Boolean Constraint Propagation.} With the
adjustments to cross-chunks of watched literals and
\Cref{inv:backtrack-compatible-watched-literals},
we adapt BCP to GB, as shown in
\Cref{alg:bcp}.
Compared to standard two-watched-literal BCP algorithms, the main
difference of \Cref{alg:bcp} comes with the handling of chunks,
cross-chunks, and skip conditions.
We note that, in practice, cross-chunks are always used in union with the chunks of
a literal. Therefore, we improve
efficiency by enforcing that $\chunks(\ell) \subseteq
\crossChunks(\ell)$ (\Cref{line:prop:reset-eta}). This way, when checking whether a clause can be
skipped during BCP, only one set
comparison is required.
In the example of \Cref{fig:motivation-conflict}, let us set the watchers of the clause $C_3$ to be $y$ and $\neg x$. When propagating $x$, we will imply $y$ with $\chunks(y) \gets \{ck_a, ck_b\}$ and update $\crossChunks(x) \gets \{ck_a, ck_b\}$ to satisfy \Cref{inv:backtrack-compatible-watched-literals}. Upon undoing $ck_a$, $x$ will be repropagated since $ck_a \in \crossChunks(x)$.


\subsection{CDCL-Based SAT Solving Using GB}
Our CDCL-based SAT solving approach using GB is summarized in
\Cref{alg:cdcl}, exploiting cost-based trail repair, conflict
analysis, backtracking, and BCP,  as presented in \Cref{sec:chunk-selection,sec:chunk:repair,sec:BCP}.
For conciseness, we only provide proof sketches here and with details in \Cref{app:soundness}.

\begin{lemma}[Graph backtracking invariant]\label{lem:gb:inv}
  The graph backtracking CDCL-based approach of \Cref{alg:cdcl} preserves \Cref{inv:backtrack-compatible-watched-literals}.
\end{lemma}
\begin{proof}
  \Cref{inv:backtrack-compatible-watched-literals} is maintained by BCP, since BCP only adds a literal $\ell$ to the propagated set $\trail$ if all clauses satisfy $\neg c_1 \in (\trail \cdot \ell) \Rightarrow \left[c_2 \in \p \land \chunks(c_2) \subseteq \bigl(\chunks(c_1) \cup \crossChunks(c_1)\bigr)\right]$. Only the clauses watched by $\neg \ell$ are checked, as they are the only ones for which $\neg c_1 \in (\trail \cdot \ell)$ can change from false to true.

  \Cref{inv:backtrack-compatible-watched-literals} is also maintained by backtracking, since it ensures $\neg c_1$ is unassigned or removed from $\trail$ before $c_2$ is unassigned.

  The other operations of \Cref{alg:cdcl} do not modify the watch lists nor the propagated set, and therefore do not violate   \Cref{inv:backtrack-compatible-watched-literals}.
\end{proof}
\begin{lemma}[Graph backtracking trail]\label{lemma:gb:conflict}
  In \Cref{alg:cdcl}, no clause is falsified by the propagated set $\trail$, or unit under $\trail$ and not satisfied by $\p$.
\end{lemma}
\begin{proof}
  We prove by contradiction. Let a clause $C$ be falsified by $\trail$ or be unit under $\trail$ and not satisfied by $\p$. Then one of its watched literals $c_1$ is falsified and propagated, and the other watched literal $c_2$ is not satisfied. This contradicts \Cref{inv:backtrack-compatible-watched-literals} and \Cref{lem:gb:inv}.
\end{proof}

\begin{theorem}[Graph backtracking]\label{thm:gb:termination}\label{thm:gb:soundness}\label{thm:gb:completeness}
\Cref{alg:cdcl} is sound, and terminating.
\end{theorem}
\begin{proof}
  Soundness follows directly from \Cref{lemma:gb:conflict}: a solution is only returned when the propagated set $\p = \trail$ and all variables are assigned. From \Cref{lemma:gb:conflict}, no clause is falsified by $\p$, and therefore the solution is correct. Learned clauses are entailed by the original formula with binary resolution. Therefore, if the formula finds a conflict without any decision, then the empty clause is entailed by the original formula, and the formula is unsatisfiable.

  The terminating backtrack candidates of \Cref{sec:chunk-selection} guarantee termination. Upon conflict, either a new clause is learned, or the most recent chunk in the conflict is undone. In the former case, there are finitely many clauses entailed by the original formula, and therefore learning a new clause is progress. In the latter case, progress must be made because of the ordering of chunks similarly to CB (\Cref{sec:chunk-selection}).
\end{proof}

\begin{algorithm}[t]
  \caption{CDCL with Graph Backtracking.}
  \label{alg:cdcl}
  \begin{algorithmic}[1]
    \State $\p = \trail = \q =  \emptyset$
    \State $\reason(\ell) = \blacksquare, \level(\ell) = \infty$
    \State \blue{$\forall \ell.\ \chunks(\ell) = \crossChunks(\ell) = \emptyset$}
    \Procedure{CDCL}{$F, \blue{\weight}$}
        \State Fill the watch lists $\wl$ for $F$
        \While {$\top$}
            \State $C \gets \Call{BCP}{ }$ \Comment{\parbox{0.45\textwidth}{\Cref{alg:bcp}}}
            \If{$C = \blacksquare$}
                \If{$|\p| = |\mathcal{V}|$}
                \Comment{\parbox{0.45\textwidth}{All variables are assigned}}
                    \State \Return \textsc{SAT}
                \EndIf
                \State $\ell \gets \Call{Decide}{ }$
                \Comment{\parbox{0.45\textwidth}{No change from (N)CB}}
                \State $\q \gets \q \cdot \ell, \level(\ell) \gets \level(\p) + 1, \blue{\chunks(\ell) \gets \{ck_\ell\}}$
                \State \Continue
            \EndIf
            \If{\blue{$\chunks(C) = \emptyset$}}
                \State \Return \textsc{UNSAT}
            \EndIf
            \State \blue{$\Gamma^* \gets \Call{SelectChunk}{C, \weight}$}
            \Comment{\parbox{0.45\textwidth}{\Cref{eq:chunk-selection}}}
            \State \blue{$D \gets \Call{Analyze}{C, \Gamma^*}$}
            \Comment{\parbox{0.45\textwidth}{\Cref{alg:conflict-analysis}}}
            \State \blue{$\Call{Undo}{\Gamma^*}$}
            \Comment{\parbox{0.45\textwidth}{\Cref{alg:backtracking}}}
            \State $c_2 \gets$ the unassigned literals in $D$
            \State $\q \gets \q \cdot c_2, \reason(c_2) \gets D, \level(c_2) \gets \level(D \setminus \{c_2\}), \blue{\chunks(c_2) \gets \chunks(D \setminus \{c_2\})}$
            \State $F \gets F \cup \{D\}$
            \If{$|D| \geq 2$}
                \State $c_1 \gets $ another literal in $D$
                \State $\wl(c_1) \gets \wl(c_1) \cup \{D\}, \wl(c_2) \gets \wl(c_2) \cup \{D\}$
                \State \blue{$\crossChunks(c_1) \gets \crossChunks(c_1) \cup \chunks(c_2)$}
            \EndIf
        \EndWhile
    \EndProcedure
\end{algorithmic}

\end{algorithm}

	\section{Simulating (N)CB with GB}
\label{sec:reduction}
In this section, we show that
graph backtracking is a generalization of both NCB and CB. Importantly, we present
how to craft an appropriate cost function $\weight$ to simulate NCB and CB with GB.
As such, GB can be used as a replacement for (N)CB, while retaining the behavior and guarantees offered by (N)CB.

Let us define a non-deterministic transition system over the state $\langle\p, \level, \reason, \phi\rangle$, where $\phi$ is the set of clauses and $\p, \level, \reason$ are as defined in \Cref{sec:preliminaries}. The transitions $T$ are the standard CDCL steps of \emph{decide}, \emph{BCP}, \emph{analyze}, and \emph{backtrack} with associated conditions, e.g., conflict analysis is only triggered when a conflict is detected by BCP.

We say that a SAT algorithm $A_1$ \emph{simulates} another SAT algorithm $A_2$ iff for any formula $F$ we have the following: any sequence of solver states $\langle\p, \level, \reason, \phi\rangle$ executing $A_1$ on $F$ is possible to be  encountered executing $A_2$ on $F$. Note that 
we do not require the same sequence of states on $A_1,A_2$, but that the same states are reachable in $A_1,A_2$. Some non-determinism in simulating SAT algorithms lies in the order of implications and conflicts due to the orders of clauses and watch lists. 

\paragraph{Simulating CB.}
We write GB$_{\weight}$ as GB parametrized with the weight function $\weight$ and define the cost function $\weight_{\text{\normalfont{CB}}}$ to be:
\begin{equation*}
  \weight_{\text{\normalfont{CB}}}(\ell) = \begin{cases}
    -1 & \text{if } \level(\ell) \geq \level(C) \\
    1 & \text{otherwise}
  \end{cases}
\end{equation*}
where $C$ is the conflicting clause.

\begin{theorem}[GB simulates CB]
  \label{lem:cb-simulation}
  GB$_{\weight_{\text{\normalfont{CB}}}}$ simulates CB.
\end{theorem}
\begin{proof}
  The minimal solution to the optimization problem \eqref{eq:chunk-selection} is the set of chunks $\Gamma^*$ whose decision levels are higher than or equal to the decision level of the conflict $C$. That is, $\bigcup_{ck \in \Gamma^*} ck = \{\ell \in \p: \level(\ell) \geq \level(C)\}$ yields the set of literals that CB would backtrack.%
  \footnote{We note that $\Gamma^*$ always satisfies the safeguard mechanism of \Cref{sec:chunk-selection}.}
  We next prove that GB$_{\weight_{\text{\normalfont{CB}}}}$ simulates CB by induction on the sequence of solver states.

  \noindent {\it Base case.} The base case $\langle\p = \emptyset, \forall \ell.\ \level(\ell) = \infty, \forall \ell.\ \reason(\ell) = \blacksquare, \phi=F\rangle$ is trivially the same for both  GB$_{\weight_{\text{\normalfont{CB}}}}$ and CB.

   \noindent {\it Step case.}   For a transition step $T\in\{decide, BCP, analyze, backtrack\}$, we show that from state $S = \langle\p, \level, \reason, \phi\rangle$, if $T_{CB}$ can reach state $S' = \langle\p', \level', \reason', \phi'\rangle$ in CB, then $T_{GB_{\weight_{\text{\normalfont{CB}}}}}$ can also reach $S'$, and vice versa.
  \begin{itemize}
    \item $decide$ is identical in GB$_{\weight_{\text{\normalfont{CB}}}}$ and CB.
    \item $BCP$ in GB$_{\weight_{\text{\normalfont{CB}}}}$ and CB differ only on the handling of watched literals, but they both find the same implications and conflicts by virtue of \Cref{inv:watched-literals}.
    Chunks $\chunks$ and cross-chunks $\crossChunks$ do not affect the states $S$. 

    \item When $\Gamma^*$ is chosen for $analyze$ in GB$_{\weight_{\text{\normalfont{CB}}}}$, it intersects with the conflict in the highest chunk in the conflict. Therefore, $\text{\normalfont{1UIP}}_{\Gamma^*}(C) = \text{\normalfont{1UIP}}_{\{ck : \level_\chunks(ck) = \level(C)\}}(C)$. Thus, $\phi' = \phi \cup \{\text{\normalfont{1UIP}}_{\Gamma^*}(C)\}$ is the same in  GB$_{\weight_{\text{\normalfont{CB}}}}$ and CB.
    \item As previously noted, the set of literals $backtrack$ed by GB$_{\weight_{\text{\normalfont{CB}}}}$ is the same as CB. Therefore, the resulting state after backtracking is the same in  GB$_{\weight_{\text{\normalfont{CB}}}}$ and CB. The solvers may differ in the propagation sets $\trail$ and $\q$, but not in their union.
  \end{itemize}
  The implication justifications $\reason$ are used identically in  GB$_{\weight_{\text{\normalfont{CB}}}}$ and CB: given the same partial assignment $\p$, both methods can find the same justifications for the same literals. Since decision levels $\level$ are defined based on decisions and justifications $\reason$, they also coincide. Therefore, for every transition $T$, if $T$ can reach state $S'$ from $S$ in  GB$_{\weight_{\text{\normalfont{CB}}}}$, then it can also be reached in CB; and vice versa. 
\end{proof}

\paragraph{Simulating NCB.} NCB is similar to CB, except that in NCB, the backjump level is not known until conflict analysis has been performed.
Let $C' = \text{\normalfont{1UIP}}_{\{ck : \level_\chunks(ck) = \level(C)\}}(C)$ be the \normalfont{1UIP} of $C$ with respect to the highest chunk in the conflict, and let $\ell'$ be the unique highest decision level literal in $C'$. The cost function $\weight_{NCB}$ is defined as:
\begin{equation*}
  \weight_{NCB}(\ell) = \begin{cases}
    -1 & \text{if } \level(\ell) > \level\left(C' \setminus \{\ell'\}\right) \\
    1 & \text{otherwise}
  \end{cases}
\end{equation*}
\begin{theorem}[GB simulates NCB]
  \label{lem:ncb-simulation}
  GB$_{\weight_{NCB}}$ simulates NCB.
\end{theorem}
\begin{proof}
  The set of backtracked chunks $\Gamma^*$ is the set whose chunk levels are higher than the second-highest decision level in the learned clause, which is exactly the set of literals that NCB backtracks. The proof then proceeds as for \Cref{lem:cb-simulation}.
\end{proof}


\section{Advanced Techniques with Graph Backtracking}\label{sec:gb:improve}
We present further improvements and details for implementing a performant GB in \Cref{alg:cdcl}.

\subsection{Chunks and Backtrack Candidates}
\label{sec:gb:implementation}
We represent chunks and cross-chunks as sparse indexed bitsets. They allow for efficient set operations by exploiting bit-level parallelism, while also being memory efficient.

While it is theoretically interesting to consider negative weights (see \Cref{sec:reduction}), $\weight$ often has its image over the positive reals. Therefore, computing the set of backtrack candidates $\Gamma$ for a conflict $C$ is not done using the powerset of chunks $2^{\chunks(\pi)}$, but rather by iterating over the chunks in the conflict only $\chunks(C)$. Indeed, if $\forall \ell.\ \weight(\ell) > 0$, then the optimal solution always has a single chunk, and computing the cost of chunk combinations is wasted. We also add a limit on the number of backtrack candidates to consider, to avoid combinatorial explosions (see \Cref{sec:chunk-merging}). In the remainder of the paper, we will only consider the case of positive weights.

\subsection{Blockers}
\label{sec:blockers}
The two-watched-literal scheme can be improved with the use of blockers~\cite{sorensson2009minisat}. A blocker is a literal $b$ in a clause $C$ stored in the watch list of a watcher $\ell$ such that if $b \in \p$, then $C$ does not need to be inspected when propagating $\neg \ell$. In NCB, \Cref{inv:watched-literals} can be relaxed to include blockers as follows.

\begin{invariant}[Watched literals with\label{inv:blockers} blockers~\cite{DBLP:conf/sat/CoutelierFK24}]
    Consider the trail \(\p = \trail \cdot \q\). For each clause $C \in F$ watched by $c_1, c_2$ and blocked by $b$ in $\wl(c_1)$, we have
    $\neg c_1 \in \trail \Rightarrow [c_2 \in \p \lor b \in \p].$
\end{invariant}

However, the same issues as discussed in \Cref{sec:BCP} arise. The order of propagation
is not necessarily compatible with the backtracking order.
Similarly to \cref{inv:backtrack-compatible-watched-literals},
we therefore strengthen \Cref{inv:blockers} with cross-chunks.
\begin{invariant}[GB blocked watched literals\label{inv:chunks}]
    Consider the trail \(\p = \trail \cdot \q\). For each clause $C \in F$ watched by $c_1, c_2$ and blocked by $b$ in $\wl(c_1)$, we have
    $$\neg c_1 \in \trail \Rightarrow \big[[c_2 \in \p \land \chunks(c_2) \subseteq \crossChunks(c_1)] \lor [b \in \p \land \chunks(b) \subseteq \crossChunks(c_1)]\big].$$
\end{invariant}


\subsection{Chunk Merging}
\label{sec:chunk-merging}
During BCP, we may discover that a clause $C$ is actually a missed implication for a decision $\ell \in ck_\ell$ such that $ck_\ell \notin \chunks(C \setminus \{\ell\})$. In case of conflict, backtracking the chunk $ck_\ell$ alone would not be fruitful since it would be reimplied by $C$ right after and the conflict would be rediscovered once more. Therefore, when backtracking $ck_\ell$, we also should backtrack one of the chunks in $\chunks(C \setminus \{\ell\})$.
We can merge the chunk $ck_\ell$ with $\chunks(C \setminus \{\ell\})$ either during propagation (eagerly), or during conflict repair (lazily).

\emph{Eager chunk merging} (ECM) is simple and similar to the handling of missed lower implications in CB \cite{DBLP:conf/sat/Nadel22}. When a clause $C$, implying a decision literal $\ell$, is discovered, we update the implication graph by reimplying $\ell$ with the clause $C$. The effect of this implication is then rippled through the implication graph, de facto eliminating the chunk $ck_\ell$ and merging it with the chunks in $\chunks(C \setminus \{\ell\})$. The trail is stably reordered to restore the topological sort.
ECM is not always optimal. For instance, let the conflict clause $C'$ spanning over both $ck_\ell$ and some chunk $ck \in \chunks(C \setminus \{\ell\})$. If $ck_\ell$ is merged eagerly, then we would backtrack $ck_\ell$ as well as $ck$. However, if the chunks were not merged, we could backtrack $ck$ alone, which is necessarily lighter than $ck_\ell \cup ck$.

On the other hand, \emph{lazy chunk merging} (LCM) virtually merges chunks during conflict repair.
To do so, during propagation, we remember potential reimplications similarly to~\cite{DBLP:conf/sat/CoutelierFK24}, and ensure that no cycle is created. When creating backtrack candidates $\Gamma$ (\Cref{sec:gb:implementation}), we enhance $\Gamma$ by expanding chunks that would be merged.
This way, the algorithm is more informed for chunk selection, and never needs to materialize the merge. It simply undoes the union of the chunks.
However, if a lot of virtual merges are introduced, the number of options to consider during conflict repair can explode and deteriorate performance.

We note $\Call{Analyze}{C,\Gamma^*}$ is adapted such that when the decision literal of a chunk is analyzed, the missed implication is used for resolution, similarly to Algorithm~4 in~\cite{DBLP:conf/sat/CoutelierFK24}.

\begin{example}
  \label{ex:gb:chunk-merging}
  In the example of \Cref{fig:motivation}, the clause $C_7 = a \lor \neg c \lor \neg d$, can merge $ck_a$ with $ck_c$ and $ck_d$. Using LCM, the backtrack options are $\Gamma = \{\{ck_a, ck_c\}, \{ck_a, ck_d\}, \{ck_c\}, \{ck_d\}\}$.
  To avoid unnecessary cost estimation, we perform subsumption simplifications on the chunks to be merged. In our example, we do not calculate the weights of $\{ck_a, ck_c\}$ and $\{ck_a, ck_d\}$, as they will always be more expensive than $\{ck_c\}$ and $\{ck_d\}$ respectively.
\end{example}


\subsection{Lightest vs. Learning Chunks}
\label{sec:backtracked-vs-learned}
\Cref{sec:chunk-selection} introduced some safeguards to ensure progress.
Namely, if a chunk does not lead to learning a new clause, and is
not at the highest level in the conflict $C$,
then we did not consider it in the terminating backtrack candidates $\Gamma$. However, this constraint can be relaxed.
If for some chunks $\Gamma_i\in \Gamma$, we manage to learn a new clause,
we can safely backtrack any chunk in $\{\Gamma_j \in 2^{\chunks(\p)}: |\Gamma_j \cap \chunks(C)| = 1\}$, even if it did not lead to learning the new clause. In particular, we can backtrack the lightest chunk.

\begin{example}
  Consider \Cref{fig:motivation-conflict}. Undoing $ck_b$ is not allowed because of the safeguard mechanism. However, in \Cref{ex:gb:chunk-merging}, $ck_b$ would be the cheapest choice because of LCM. We learn a new clause $\text{\normalfont{1UIP}}_{ck_d}(C_8) = C_9 = \neg w \lor \neg y \lor \neg z$ and backtrack $ck_b$.
\end{example}

	\section{Empirical Evaluation}
\label{sec:experiments}
We implemented our graph backtracking approach in the experimental solver \napsat~\cite{napsat}. Our solver is based on the CDCL algorithm~\cite{DBLP:journals/tc/Marques-SilvaS99} using the two-watched-literal scheme~\cite{DBLP:conf/dac/MoskewiczMZZM01}, VSIDS branching heuristic~\cite{DBLP:conf/dac/MoskewiczMZZM01} with phase caching \cite{DBLP:conf/sat/EenS03} and clause shrinking~\cite{DBLP:conf/sat/SorenssonB09}. The core implementation consists in roughly 2.7k effective lines of C++ code, covering all options of GB presented in \Cref{sec:details,sec:gb:improve}, as well as standard (N)CB. For a fair comparison, we attempted to maximize code sharing between the different backtracking options;
as such, only 2\% of the code of NCB is not shared with GB, whereas GB represents around 27\% more code than NCB\footnote{Calculated using \texttt{gcc --coverage} with \texttt{lcov}}.
Since our focus is on the comparison of backtracking strategies, we did not implement auxiliary techniques, such as pre- and in-processing, clause minimization, and non-redundant clause deletion (we only purge clauses satisfied at root level). To eliminate engineering differences, we compare the different backtracking strategies in the same solver.
All our computations and results were executed on an AMD EPYC 7502 2$\times$32-core processor running at 2.5~GHz and equipped with 1~TB of RAM.

\paragraph{Weight heuristic.}
We used a simple weight heuristic to minimize propagations. If a literal $\ell$ is propagated, then $\weight(\ell) = 8$; otherwise, $\weight(\ell) = 1$.
The coefficient $8$ was chosen without any fine-tuning. A power of two is however beneficial for stability of floating point arithmetic.
We also add a small penalty to chunks whose decisions are located earlier in the trail. This encourages the solver to behave more like CB when the weights only differ slightly.

\paragraph{Configurations.}
Below, we present the configurations that were evaluated. \napsat{} always uses blockers (\Cref{sec:blockers}).
\begin{itemize}[leftmargin = 2.4\parindent,labelindent=2\parindent]
	\item[NCB:] Non-chronological backtracking as described in \Cref{sec:related} and \cite{DBLP:conf/sat/EenS03}.
	\item[CB:] Chronological backtracking as described in \Cref{sec:related} and \cite{DBLP:conf/sat/NadelR18}.
	\item[LSCB:]  Lazy Strong Chronological Backtracking as described in~\cite{DBLP:conf/sat/CoutelierFK24}.
	\item[GB:] Vanilla graph backtracking as described in \Cref{sec:details}.
	\item[LCM:] Lazy chunk merging as described in \Cref{sec:chunk-merging}.
	\item[ECM:] Eager chunk merging as described in \Cref{sec:chunk-merging}.
	\item[BB:] Backtracking the best chunk as described in \Cref{sec:backtracked-vs-learned}.
\end{itemize}

\paragraph{Benchmarks.}
To illustrate the performance of graph backtracking, we generated 1000 satisfiable instances of 3-coloring problems using the \texttt{cnfgen} tool~\cite{DBLP:conf/sat/LauriaENV17} with arguments \texttt{kcolor 3 gnm 650 1469}. We aimed at finding satisfiable instances sufficiently difficult for our solver, while not too difficult to solve within a reasonable time. Each problem is run with a timeout of 2 hours.
Since the SAT competition problems are tuned to existing SAT techniques, they are not the best fit for our experiments. Graph coloring problems are well-suited for our application because they do not involve too many variables. Further, we could tune the hardness of problems to the capabilities of our solver, which is important to get a good picture of the performance of the different backtracking strategies.

\begin{figure}
	\centering
  \includegraphics[width=\textwidth]{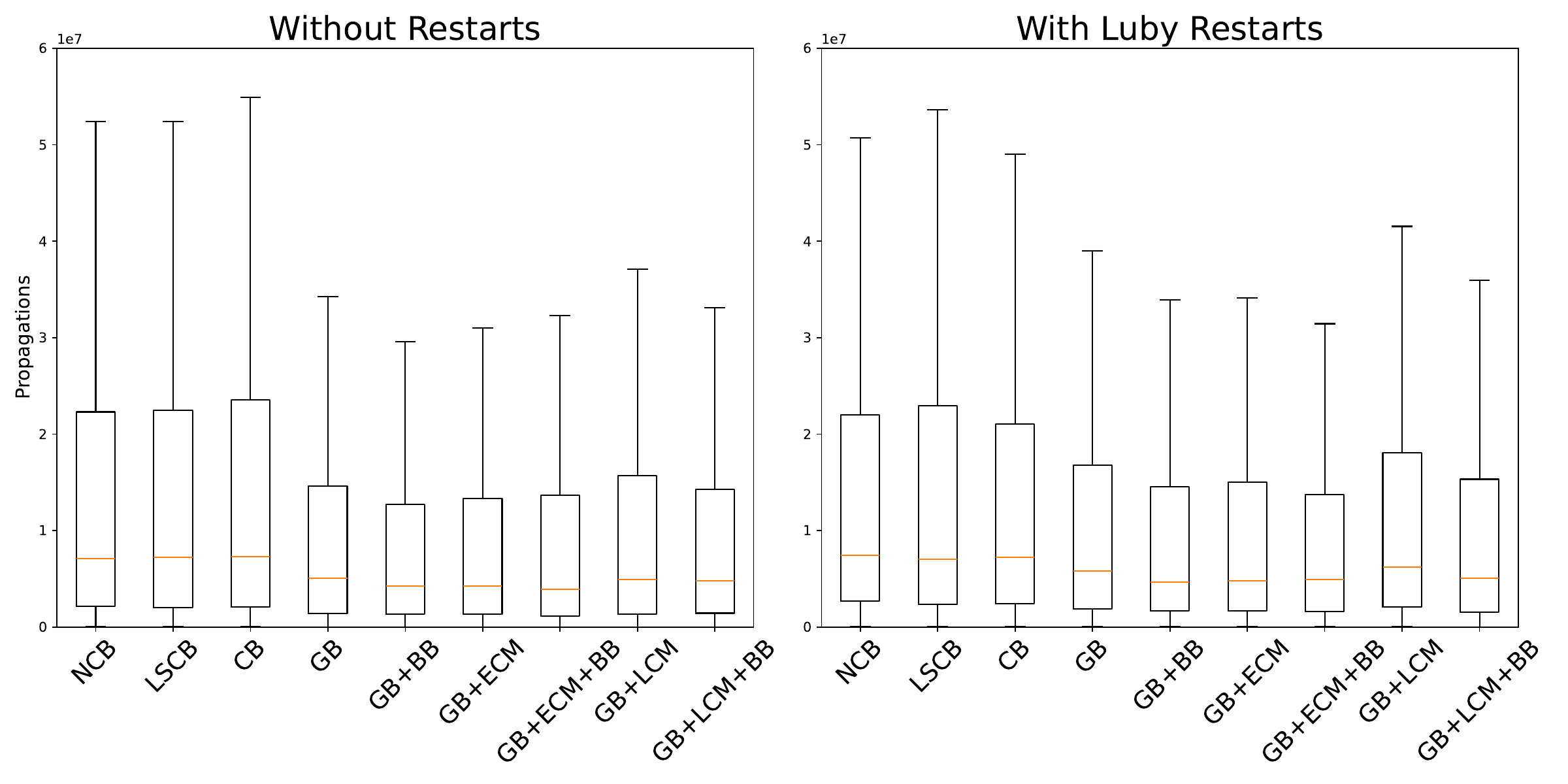}
	\caption{Comparison of the number of propagations for NCB, CB and GB variants.}
	\label{fig:results}
\end{figure}

\begin{table}
  \centering
  \begin{tabular}{l@{\hskip 0.1in}|c|cc|@{\hskip 0.1in}lc@{\hskip 0.1in}|@{\hskip 0.1in}lc}
    \toprule
    Option & Timeout & PAR2 & & Time (s) & & \multicolumn{2}{l}{Propagations $\times10^6$} \\
    \midrule
    NCB         & 16 & 380.7 & -       & 106  $\pm$ 428 & -       & 23.7 $\pm$ 46.4 & -       \\
    LSCB        & 12 & 335.0 & 88.00\% & 100  $\pm$ 388 & 94.27\% & 22.1 $\pm$ 41.0 & 93.07\% \\
    CB          & 11 & 333.1 & 87.50\% & 103  $\pm$ 414 & 97.33\% & 22.9 $\pm$ 42.8 & 96.30\% \\
    \midrule
    GB          & 17 & 378.1 & 99.31\% & 103  $\pm$ 402 & 96.91\% & 14.2 $\pm$ 25.3 & 60.05\% \\
    GB+BB       & 12 & 304.2 & 79.71\% & 91.6 $\pm$ 380 & 85.98\% & 12.6 $\pm$ 23.0 & 53.27\% \\
    GB+ECM      & 11 & 299.7 & 78.72\% & 72.4 $\pm$ 310 & 67.89\% & 12.6 $\pm$ 23.8 & 52.96\% \\
    GB+ECM+BB   & 12 & 303.5 & 79.72\% & 84.9 $\pm$ 366 & 79.66\% & 13.0 $\pm$ 25.0 & 54.64\% \\
    GB+LCM      & 14 & 376.8 & 98.98\% & 124~ $\pm$ 485 & 117.1\% & 15.8 $\pm$ 29.9 & 66.39\% \\
    GB+LCM+BB   & 16 & 387.4 & 101.8\% & 97.0 $\pm$ 332 & 91.05\% & 13.9 $\pm$ 24.2 & 58.71\% \\    \bottomrule
  \end{tabular}
  \caption{Statistics of the experiments. The timeout column indicates the number of instances that timed out for each configuration. PAR2 indicates the penalized average of runtimes, where each timeout is counted as 2 times the timeout value (2 hours). The average times and propagations are computed excluding the 50/1000 instances for which at least one of the configurations timed out.
  Standard deviation is shown after the $\pm$ symbol. Relative values (in percentages) are computed with respect to NCB and displayed on the right of the respective columns.}
  \label{tab:results-condensed}
\end{table}

\paragraph{Result analysis.}
Our experiments are summarized in \Cref{tab:results-condensed} and \Cref{fig:results} with further details in \Cref{sec:app:experiments}. Our approach consistently executes fewer propagations than (N)CB. Our best configuration is GB+ECM as it executes 47\% fewer propagations than NCB and decreases the runtime by around 30\%.

In our experiments, ECM makes the solver more stable, reaching fewer timeouts and reducing the variance of the runtime. This is because, without chunk merging, chains of conflicts can occur and lead to undoing the same number of literals in more backtracking steps, which is very costly. On the other hand, LCM introduces a lot of variance in the cost of computing backtrack candidates and their weights. This is why LCM has the lowest propagations per second, but also the highest variance in runtime. While the results of our experiments suggest that LCM is not very effective, it is still worth investigating it further. Problems with many variables and binary clauses may benefit LCM because the cost of eager merging is proportional to the number of literals that need to be reimplied.

Backtracking the best chunk (BB) has mixed results, improving the performance of GB alone, and with LCM, but not with ECM. However, the penalty witnessed with ECM is not very large. It seems that BB is a good option because it balances the safeguard mechanism described in \Cref{sec:chunk-selection}. This mechanism is meant to prevent a very unlikely scenario where the solver would loop forever. Sacrificing the best chunk for learning a clause goes against the philosophy of GB. BB gets us back to the core idea of GB, while keeping termination.

Unsurprisingly, restarts penalize GB because they clash with the intention of GB to minimize the number of unassignments. The comparison can be found on \Cref{fig:results}. Indeed, restarts work well when propagations are cheap, which is less the case in GB. Reordering the decisions is less important in GB since the solver can backtrack any of the decisions involved in the conflict, and not only the most recent one.

The results also highlight the trade-off between time and number of propagations. While GB executes significantly fewer propagations, each of them is slower than in (N)CB. This is because of the chunk and cross-chunk maintenance, and the more complex backtracking procedure, which requires going through the entire trail to find the literals that need to be propagated again. GB shines particularly when the weight of chunks varies significantly.


\section{Conclusion and Future Work}
We introduce a graph backtracking (GB), yielding a more surgical backtracking approach for CDCL solvers. 
GB generalizes the backtracking process of (N)CB and adapts to user preferences over literal unassignments. Our GB approach is sound and terminating. 
We implemented GB in the experimental solver \napsat and showed that it performs better on graph coloring benchmarks. Particularly, GB backtracks fewer literals in single-shot SMT conflicts, compared to (N)CB.

The AVATAR framework~\cite{DBLP:conf/cav/Voronkov14} is a promising application for graph backtracking as it
does not require a given order of literals in the trail, only sets of assigned literals.
The integration of \napsat in the first-order prover \vampire{}~\cite{DBLP:conf/cav/BartekBCHHHKRRRSSV25} is ongoing. Scaling the algorithms to large number of decisions, and therefore chunks is still an open problem.
Integrating GB in SMT solving is another future work direction.
SMT solvers typically follow a more rigid assignment structure for their theory solvers,
which does not fit well with graph backtracking. Relaxing this stack assumption to benefit from chronological and graph backtracking is future work.

	\bibliographystyle{plainurl}
	\bibliography{refs}

	\newpage
	\appendix

\section{Soundness, Completeness and Termination}
\label{app:soundness}
We prove that the graph backtracking approach of \Cref{alg:cdcl}  is sound, complete, and terminating.

\begin{dupllemma}{\ref{lem:gb:inv}}
  The graph backtracking CDCL-based approach of \Cref{alg:cdcl} preserves \Cref{inv:backtrack-compatible-watched-literals}.
\end{dupllemma}
\begin{proof}
  The proof is similar to the proof of soundness of lazy strong chronological back-tracking~\cite{DBLP:conf/sat/CoutelierFK24}. We show that \Cref{inv:backtrack-compatible-watched-literals} is preserved by each of the procedures of GB. The detailed steps for BCP are included in the source code of \napsat.

  \paragraph{BCP.}
  The BCP algorithm transfers, one by one, literals from the waiting queue $\q$ to the propagated set $\trail$.
  $\Call{BCP}$ preserves \Cref{inv:backtrack-compatible-watched-literals} by induction.
  At the start, $\p = \trail = \q = \emptyset$ and the invariant is trivially satisfied.
  Let us assume the current partial assignment $\p = \trail \cdot \q$ satisfies \Cref{inv:backtrack-compatible-watched-literals} before executing $\Call{BCP}$.
  A literal $\ell \in \q$ can be added to the propagated set $\trail$ if we first ensure that for each clause $C \in F$ watched by $c_1, c_2$,
  \begin{equation}
    \neg c_1 \in (\trail \cdot \ell) \Rightarrow [c_2 \in \p \land \chunks(c_2) \subseteq \crossChunks(c_1)]
    \label{eq:invariant-step}
  \end{equation}

  All clauses that are not watched by $\neg \ell$ satisfy \Cref{eq:invariant-step} by the inductive hypothesis, and do not need to be inspected. We therefore only consider clauses watched by $\neg \ell$.
  When we attempt to add $\ell$ to $\trail$, for each clause $C$, one of following can occur:
  \begin{itemize}
    \item If $c_2 \in \p \land \chunks(c_2) \subseteq \crossChunks(c_1)$, then nothing needs to be done.
    \item If there exists a replacement literal $r \in C \setminus \{c_2\}$ such that $\neg r \notin \p$, we watch $r$ instead of $c_1$ and have $\neg r \notin (\trail \cdot \ell)$ since $\ell \in \q \subseteq \p$.
    \item If $\neg c_2 \in \p$ and no good replacement exists, then $\Call{BCP}$ returns the conflict clause $C$ and $\ell$ is not added to $\trail$.
    \item If $c_2$ is unassigned and no good replacement exists, then $c_2$ is added to $\q$ such that $\chunks(c_2) = \chunks(C)$.
    We then update $\crossChunks(c_1)$ to include $\chunks(c_2)$.
    \item If $c_2 \in \p$ but $\chunks(c_2) \nsubseteq \crossChunks(c_1)$, then the algorithm forces $\chunks(c_2) \subseteq \crossChunks(c_1)$ by updating $\crossChunks(c_1)$.
  \end{itemize}

  The cases discussed above cover all possible scenarios. For each, either $\ell$ is not added to $\trail$ because of a conflict, or \Cref{eq:invariant-step} is satisfied.

  \paragraph{Chunk Selection \& Conflict Analysis.}
  Chunk Selection and Conflict analysis do not change the state of the partial assignment, nor the clauses. Therefore, they trivially preserve the invariant.

  \paragraph{Backtracking.}
  Backtracking may change the partial assignment by removing literals. However, if a clause watched by $c_1, c_2$ was such that $\neg c_1 \in \trail \land c_2 \in \p$, then either $c_2$ is left untouched, or $\neg c_1$ is removed from $\trail$ (either removed from $\p$ or placed in $\q$). Therefore, after backtracking, it is not possible that $\neg c_1 \in \trail \Rightarrow c_2 \in \p$ is violated.
  The chunk sets of literals remain unchanged, and therefore the invariant is preserved.

  \paragraph{CDCL.}
  Only two aspects of CDCL have not been examined so far. First, the decisions made during the search. However, since this part of the algorithm only adds literals to the waiting queue $\q$, the invariant is trivially preserved.
  Second, when the learned clause is added to the clause set $F$, the invariant is restored because we know that the new clause has exactly one unassigned literal (by design of the UIP), and therefore, pushing the new implication and updating the cross-chunks of the other literal is sufficient to maintain the invariant.
\end{proof}

Strong with \Cref{lem:gb:inv}, we can now show that the graph backtracking algorithm find all implications and cannot miss a conflict.

\begin{dupllemma}{\ref{lemma:gb:conflict}}
  In \Cref{alg:cdcl}, no clause is conflicting or unit with the propagated set $\trail$ and not satisfied by $\p$.
\end{dupllemma}
\begin{proof}
  By contradiction, if a clause $C$ was falsified or unit by the propagated set $\trail$, then at least one of its watchers would be falsified in $\trail$, and the other cannot be satisfied. This contradicts \Cref{lem:gb:inv}.
\end{proof}

\begin{dupltheorem}{\ref{thm:gb:soundness}.1}
CDCL-based SAT solving using GB in \Cref{alg:cdcl} is sound.
\end{dupltheorem}
\begin{proof}
  With \Cref{lemma:gb:conflict}, a satisfiying assignment cannot be conflicting with any clause of the formula. Therefore, if GB CDCL returns SAT, the assignment is indeed satisfying.
  If GB CDCL returns UNSAT, then the empty clause was derived by resolution from clauses of the formula. Therefore, the formula is indeed unsatisfiable.
\end{proof}

\begin{dupltheorem}{\ref{thm:gb:completeness}.2}
  CDCL-based SAT solving using GB in \Cref{alg:cdcl} is complete.
\end{dupltheorem}
\begin{proof}
  If the formula is satisfiable, then there exists a satisfying assignment $\p$. Since the algorithm only learns clauses that are entailed by the formula, $\p$ remains a satisfying assignment of the formula throughout the execution of the algorithm. Therefore, the algorithm cannot return UNSAT.
  If the formula is unsatisfiable, then by completeness of resolution, the empty clause can be derived by resolution from clauses of the formula. Therefore, the algorithm eventually returns UNSAT.
\end{proof}

\begin{dupltheorem}{\ref{thm:gb:termination}.3}
CDCL-based SAT solving using GB in \Cref{alg:cdcl} terminates.
\end{dupltheorem}
\begin{proof}
  The termination of each subalgorithm of GB CDCL follows from the finite set of variables and clauses. It does not need further explanation as it is similar to other CDCL variants.

  Progress is made by two means: (i)~Assigning literals that do not belong to a chunk (root level in NCB). When a literal $\ell$ is assigned and does not belong to any chunk, it means it is a direct implication of a subset of the formula $F$. Therefore, it cannot be backtracked. Since there are finitely many variables, assigning a literal $\ell$ such that $\chunks(\ell) = \emptyset$ is progress towards termination.
  (ii)~Learning a new entailed clause makes progress since there are finitely many distinct clauses entailed by the formula $F$.

  We show that always eventually, if the algorithm is not finished, it will make progress.
  When a conflict $C$ is discovered, we perform conflict analysis and obtain the entailed clause $C'$ one of two things happens:
  \begin{enumerate}
    \item $C' \notin F$ and we make progress by virtue of (ii)
    \item $C' \in F$ and the learned clause is redundant. Let $\chunks(C)$ be the set of chunks on which $C'$ is conflicting.
    \begin{enumerate}
      \item If $|\chunks(C')| = 1$, then, ignoring root level literals, the learned clause must be of length one as well because of the UIP algorithm. After backtracking, we imply a literal $\ell$ such that $\chunks(\ell) = \emptyset$ and make progress by virtue of (i).
      \item If $|\chunks(C')| > 1$, then the algorithm undoes the chunk $ck \in \chunks(C')$ whose decision $d$ is at the highest decision level among $\chunks(C')$ (the one that was decided last). This is enforced by the safeguard of \Cref{sec:chunk-selection}. Let the literal $\ell \in ck$ be the UIP. We call \emph{hot potatoes} the literals in a chunk that were flipped because of a conflict. As long as a potato lives, the conflict clause that created it cannot happen again (since the clause is satisfied). Hot potatoes might get undone, but only if a lower one is added. Hot potatoes can only sink in terms of decision level. Therefore, it cannot be that the same conflicts are infinitely rediscovered without making progress.
    \end{enumerate}
  \end{enumerate}
\end{proof}

	\section{Conflict Search Strategy}
\label{sec:conflict-strategy}
When a conflict is detected, a few options are possible. The simplest,
and most commonly used, is to interrupt BCP and trigger conflict
repair immediately.
However, this is not necessarily always the best option.
Indeed, we sometimes want to gather all conflicts that can be
discovered at once,
especially when the trail does not have a monotonic structure like in CB or GB~\cite{DBLP:conf/sat/Nadel22,briefs-fleury-2025}.

We differentiate between three strategies: (i) immediate, (ii)
partial, and (iii) exhaustive conflict repair. (i) Immediate conflict
repair is the simplest strategy.
As soon as a conflict is detected, BCP is interrupted and conflict
analysis is triggered.
(ii) Partial conflict repair propagates all literals in
the first conflict clause before triggering conflict analysis.
Finally, (iii) exhaustive conflict repair continues BCP until the queue $\q$ is empty.

State-of-the-art solvers like Kissat~\cite{BiereFleuryPolitt-SAT-Competition-2023-solvers} and CaDiCaL~\cite{DBLP:conf/cav/BiereFFFFP24} do immediate conflict repair, with Intel SAT~\cite{DBLP:conf/sat/Nadel22} doing exhaustive.

\paragraph{Immediate Conflict Repair.}
In light of the termination criterion presented in \Cref{sec:chunk-selection}, it is important to learn at least one new clause for each conflict repair that does not occur at the highest chunk. However, a hidden problem occurring in both GB and CB is the possibility to learn an existing, or subsumed clause. This can happen if multiple conflicts occur simultaneously. Using immediate conflict repair, we analyze the first conflict $C$ to obtain $D$. If $\exists C' \in F$ such that $C' \subseteq D$, we would not make progress. A brutish way to verify whether the clause $D$ is new, is to traverse the watch lists of each literal of $D$. If the clause is subsumed, another chunk must be selected.

\paragraph{Partial and Exhaustive Conflict Repair.}
While no known attempt was successful in finding value in exhaustive conflict repair in NCB, it has been experimented on in the context of CB~\cite{DBLP:conf/sat/Nadel22,briefs-fleury-2025}. In applications where the procedure calling the SAT solver is orders of magnitude more expensive than the SAT solver itself, it might be worth gathering more conflicts to select the lightest set of literals with more accuracy.
Collecting more conflicts is transparent to the user and allows better choices of chunk to backtrack.

By virtue of the topological order of propagations, partial conflict repair ensures that if the first conflict's learned clause $D$ is subsumed by another existing clause $C'$, then $C'$ is conflicting and will also be discovered. Exhaustive conflict repair ensures that if any learned clause is subsumed, the subsuming clause will part of the conflicting clause set.

Partial and exhaustive conflict repair are non-trivial to implement. Gathering all the conflicts is not difficult, but handling them simultaneously is. We only present the exhaustive conflict repair procedure in broad strokes, with details in the code base. Partial conflict repair procedure is similar, but with weaker guarantees on the novelty of learned clauses.

Considering a set $\kappa$ of conflicts discovered during BCP, the first step is to find the set $\Gamma = \{\Gamma_i | \forall C \in \kappa. \Gamma_i \cap \chunks(C) \neq \emptyset\}$ of all set of chunks $\Gamma_i$ that can be backtracked to repair all conflicts at once. We use a greedy algorithm enhanced with subsumption to find a small set of solutions.
Then, we expand each $\Gamma_i$ with the lazy merges (\Cref{sec:chunk-merging}) and perform another subsumption simplification.

Note that this definition of backtrack candidates after subsumption ensures that at least one of the conflicts can either learn a UIP, or is already a UIP itself. Therefore, we are guaranteed to be able to imply a flipped literal after backtracking, which is necessary to ensure progress.

We filter all sets of chunks $\Gamma_i$ that cannot lead to new clauses and do not contain the highest chunk level out of $\Gamma$. We then select $\Gamma^*$, the lightest chunks in $\Gamma$ with respect to $\weight$.

Finally, we attempt to learn new clauses for each conflict $C \in \kappa$ using the selected $\Gamma^*$. If all learned clauses are subsumed by existing clauses, we select another $\Gamma_i$ and try again. We stop when we learn a new clause, or when we analyze a top level chunk.

	\section{Additional Empirical Analysis}
\label{sec:app:experiments}

In this section, we provide additional representation and comments of our experimental results. The option acronyms and experimental setup were defined and described in \Cref{sec:experiments}.

Additionally, we evaluated the different conflict research strategies as descived in \Cref{sec:conflict-strategy} and the impact of restarts on GB. These results are available in \Cref{sec:app:experiments}.

However, due to the large number of configurations and cost of exhaustive conflict search, the following results were obtained on smaller problems than the ones presented in \Cref{sec:experiments}. We generated 1000 satisfiable instances of 3-coloring problems using the \texttt{cnfgen} tool~\cite{DBLP:conf/sat/LauriaENV17} with arguments \texttt{kcolor 3 gnm 400 920}.

\subsection{Experimental Setup}

\paragraph{Target metric.}
The motivation of GB is to give some control to choose which literals should be unassigned upon conflict discovery. To illustrate the performance of GB, we use an approximation of the work performed by a user solving incremental SAT instances. We measure the number of \emph{synchronizations} (syncs) between the solver and the user. That is, before each decision of the SAT solver, we assume a user synchronizing its internal state with the SAT solver (similar to user propagators~\cite{DBLP:journals/jair/FazekasNPKSB24} and how \verit~\cite{DBLP:conf/cade/BoutonODF09} works).
If a variable has changed polarity in the assignment compared to the last synchronization, we count one \texttt{sync}.
To minimize the number of synchronizations, the decision heuristic assigns literals to the same polarity as the last synchronization.

\paragraph{Decision polarity.}
In standard mode, \napsat{} uses phase cashing to assign decision literals to the same polarity as its last assignment. To minimize the changes between synchronizations, we replace phase caching with a more aggressive approach that assigns decision literals to the same polarity as the last synchronization. This way, the solver prefers to keep literals synchronized. In addition, we also add a small penalty to chunks whose decisions are located earlier in the trail. The encourages the solver to behave more like CB when it does not matter. Further, experiments show that analyzing conflicts on earlier chunks generally leads to larger learned clauses.

\paragraph{Additional options.}
We also evaluated the impact of the conflict strategy (immediate, partial and exhaustive conflict repair) and the impact of restarts on GB. In addition to the options described in \Cref{sec:experiments}, we also evaluated the following configurations:
\begin{itemize}[leftmargin = 2\parindent,labelindent=2\parindent]
  \item[PCR] Partial conflict repair.
  \item[ECR] Exhaustive conflict repair.
\end{itemize}
In terms of synchronizations, LSCB and CB are very similar, we therefore did not run LSCB for these experiments.
No timeout was used for these experiments.

\subsection{Results and Analysis}

\paragraph{Raw statistics.} \Cref{tab:results-raw} shows the average time, number of synchronizations for each configuration. Graph backtracking (GB), no matter its variant, consistently executes fewer synchronizations than NCB and CB, on average. Our best configuration, GB+BB, executes around 37\% fewer synchronizations than NCB, with a reasonable overhead in execution time. ECR can bring this number to 44\%, but at a significant cost in terms of execution time and number of propagations.

\begin{table}
  \centering
  \begin{tabular}{l@{\hskip 0.1in}|@{\hskip 0.1in}lr@{\hskip 0.1in}|@{\hskip 0.1in}lr}
    \toprule
    Option & Time (s)  & Time relative & Sync $\times10^3$  & Sync relative \\
    \midrule
    NCB           & 1.43 $\pm$ 4.66 & - ~~~~~~~~~~    & 85.1 $\pm$ 118 & -~~~~~~~~~~ \\
    CB            & 1.26 $\pm$ 3.26 & 88.09\% (88.09\%) & 74.0 $\pm$ 98.1 & 86.87\% (86.87\%) \\
    GB            & 2.24 $\pm$ 6.35 & 156.2\% (156.2\%) & 49.4 $\pm$ 75.6 & 57.99\% (57.99\%) \\
    GB+BB         & 3.16 $\pm$ 20.3 & 220.9\% (220.9\%) & 53.4 $\pm$ 106 & 62.77\% (62.77\%) \\
    GB+ECM        & 2.27 $\pm$ 10.8 & 158.6\% (158.6\%) & 55.7 $\pm$ 96.4 & 65.38\% (65.38\%) \\
    GB+ECM+BB     & 1.82 $\pm$ 5.20 & 126.8\% (126.8\%) & 52.0 $\pm$ 79.8 & 61.08\% (61.08\%) \\
    GB+LCM        & 2.54 $\pm$ 6.61 & 177.1\% (177.1\%) & 54.4 $\pm$ 82.6 & 63.96\% (63.96\%) \\
    GB+LCM+BB     & 3.09 $\pm$ 12.6 & 215.7\% (215.7\%) & 56.2 $\pm$ 99.9 & 66.00\% (66.00\%) \\
    \midrule
    NCB+PCR       & 1.71 $\pm$ 4.75 & - ~~~~  (119.1\%) & 83.1 $\pm$ 110  & - ~~~~  (97.59\%) \\
    CB+PCR        & 1.65 $\pm$ 4.77 & 96.67\% (115.2\%) & 73.0 $\pm$ 97.0 & 87.88\% (85.76\%) \\
    GB+PCR        & 3.17 $\pm$ 11.0 & 185.8\% (221.3\%) & 49.7 $\pm$ 79.4 & 59.82\% (58.38\%) \\
    GB+BB+PCR     & 3.41 $\pm$ 11.6 & 199.6\% (237.8\%) & 51.6 $\pm$ 83.0 & 62.12\% (60.62\%) \\
    GB+ECM+PCR    & 2.80 $\pm$ 13.3 & 164.3\% (195.7\%) & 52.4 $\pm$ 88.4 & 63.10\% (61.58\%) \\
    GB+ECM+BB+PCR & 2.49 $\pm$ 8.50 & 146.0\% (174.0\%) & 53.1 $\pm$ 82.8 & 63.95\% (62.41\%) \\
    GB+LCM+PCR    & 3.40 $\pm$ 12.3 & 199.2\% (237.3\%) & 52.7 $\pm$ 82.1 & 63.38\% (61.86\%) \\
    GB+LCM+BB+PCR & 5.73 $\pm$ 41.2 & 335.8\% (400.0\%) & 58.3 $\pm$ 113  & 70.16\% (68.47\%) \\
    \midrule
    NCB+ECR       & 27.3 $\pm$ 89.3 & - ~~~~  (1909\%) & 78.6 $\pm$ 106  & - ~~~~  (92.37\%) \\
    CB+ECR        & 27.3 $\pm$ 102  & 99.73\% (1903\%) & 66.0 $\pm$ 89.4 & 83.95\% (77.54\%) \\
    GB+ECR        & 74.8 $\pm$ 589  & 273.5\% (5221\%) & 44.3 $\pm$ 77.6 & 56.38\% (52.07\%) \\
    GB+BB+ECR     & 72.3 $\pm$ 853  & 264.3\% (5046\%) & 43.2 $\pm$ 74.8 & 54.97\% (50.77\%) \\
    GB+ECM+ECR    & 51.5 $\pm$ 655  & 188.2\% (3593\%) & 48.4 $\pm$ 87.5 & 61.52\% (56.82\%) \\
    GB+ECM+BB+ECR & 47.5 $\pm$ 263  & 173.7\% (3316\%) & 49.1 $\pm$ 82.6 & 62.47\% (57.70\%) \\
    GB+LCM+ECR    & 68.8 $\pm$ 385  & 251.7\% (4804\%) & 50.1 $\pm$ 84.3 & 63.66\% (58.80\%) \\
    GB+LCM+BB+ECR & 50.7 $\pm$ 172  & 185.4\% (3539\%) & 47.2 $\pm$ 70.6 & 59.97\% (55.39\%) \\
    \bottomrule
  \end{tabular}
  \caption{Average time and number of synchronizations for each configuration with \textbf{restarts disabled}.
  Standard deviation is shown after the $\pm$ symbol. The relative values are computed with respect to the NCB configuration with the same conflict research method and restart policy. In parentheses, we also show the relative values with respect to the NCB with immediate conflict repair and no restart.}
  \label{tab:results-raw}
\end{table}
\begin{table}
  \centering
  \begin{tabular}{l@{\hskip 0.1in}|@{\hskip 0.1in}lr@{\hskip 0.1in}|@{\hskip 0.1in}lr}
    \toprule
    Option & Time (s)  & Time relative & Sync $\times10^3$  & Sync relative \\
    \midrule
    NCB           & 1.82 $\pm$ 7.11 & - ~~~~  (127.2\%) & 85.1 $\pm$ 135  & - ~~~~  (99.98\%) \\
    CB            & 1.68 $\pm$ 5.48 & 92.28\% (117.4\%) & 83.0 $\pm$ 122  & 97.56\% (97.55\%) \\
    GB            & 3.53 $\pm$ 12.2 & 193.7\% (246.4\%) & 63.0 $\pm$ 99.9 & 73.97\% (73.96\%) \\
    GB+BB         & 3.63 $\pm$ 20.3 & 199.5\% (253.8\%) & 63.7 $\pm$ 105  & 74.81\% (74.81\%) \\
    GB+ECM        & 2.88 $\pm$ 9.80 & 158.0\% (201.1\%) & 66.4 $\pm$ 110  & 78.00\% (77.99\%) \\
    GB+ECM+BB     & 3.41 $\pm$ 20.2 & 187.2\% (238.2\%) & 66.7 $\pm$ 125  & 78.39\% (78.39\%) \\
    GB+LCM        & 4.66 $\pm$ 26.2 & 256.0\% (325.7\%) & 66.0 $\pm$ 126  & 77.58\% (77.57\%) \\
    GB+LCM+BB     & 4.23 $\pm$ 16.9 & 231.9\% (295.0\%) & 66.6 $\pm$ 116  & 78.28\% (78.26\%) \\
    \midrule
    NCB+PCR       & 2.56 $\pm$ 16.5 & - ~~~~  (179.0\%) & 82.3 $\pm$ 130 & - ~~~~  (96.70\%) \\
    CB+PCR        & 1.96 $\pm$ 6.86 & 76.55\% (137.0\%) & 75.4 $\pm$ 107 & 91.54\% (88.52\%) \\
    GB+PCR        & 8.15 $\pm$ 109  & 317.8\% (568.9\%) & 64.3 $\pm$ 117 & 78.15\% (75.57\%) \\
    GB+BB+PCR     & 5.61 $\pm$ 33.3 & 218.7\% (391.4\%) & 62.0 $\pm$ 103 & 75.28\% (72.79\%) \\
    GB+ECM+PCR    & 4.17 $\pm$ 21.4 & 162.6\% (291.1\%) & 65.6 $\pm$ 110 & 79.64\% (77.02\%) \\
    GB+ECM+BB+PCR & 4.85 $\pm$ 35.3 & 189.3\% (338.8\%) & 65.5 $\pm$ 116 & 79.54\% (76.92\%) \\
    GB+LCM+PCR    & 8.18 $\pm$ 86.5 & 319.0\% (570.9\%) & 67.4 $\pm$ 123 & 81.86\% (79.16\%) \\
    GB+LCM+BB+PCR & 5.96 $\pm$ 31.5 & 232.6\% (416.3\%) & 66.0 $\pm$ 110 & 80.14\% (77.50\%) \\
    \midrule
    NCB+ECR       & 38.7 $\pm$ 162  & - ~~~~  (2702\%)  & 78.4 $\pm$ 118  & - ~~~~  (92.02\%) \\
    CB+ECR        & 30.6 $\pm$ 101  & 78.93\% (2132\%)  & 68.9 $\pm$ 93.1 & 87.99\% (80.97\%) \\
    GB+ECR        & 90.8 $\pm$ 346  & 234.5\% (6336\%)  & 51.6 $\pm$ 75.9 & 65.87\% (60.62\%) \\
    GB+BB+ECR     & 102  $\pm$ 463  & 263.7\% (7126\%)  & 51.3 $\pm$ 81.3 & 65.50\% (60.28\%) \\
    GB+ECM+ECR    & 94.4 $\pm$ 1329 & 243.8\% (6588\%)  & 58.5 $\pm$ 118  & 74.71\% (68.75\%) \\
    GB+ECM+BB+ECR & 74.4 $\pm$ 651  & 192.3\% (5196\%)  & 59.5 $\pm$ 104  & 75.98\% (69.92\%) \\
    GB+LCM+ECR    & 181  $\pm$ 1619 & 468.8\% (12669\%) & 58.2 $\pm$ 111  & 74.24\% (68.32\%) \\
    GB+LCM+BB+ECR & 129  $\pm$ 684  & 333.7\% (9017\%)  & 54.6 $\pm$ 89.8 & 69.62\% (64.07\%) \\
    \bottomrule
  \end{tabular}

  \caption{Average time and number of synchronizations for each configuration with \textbf{restarts enabled}.
  Standard deviation is shown after the $\pm$ symbol. The relative values are computed with respect to the NCB configuration with the same conflict research method and restart policy. In parentheses, we also show the relative values with respect to the NCB with immediate conflict repair and no restart.}
  \label{tab:results-raw-restart}
\end{table}

\begin{table}
  \centering
  \begin{tabular}{l@{\hskip 0.1in}|c|cc|@{\hskip 0.1in}lc@{\hskip 0.1in}|@{\hskip 0.1in}lc}
    \toprule
    Option & Timeout & PAR2 & & Time (s) & & \multicolumn{2}{l}{Propagations $\times10^6$} \\
    \midrule
    NCB       & 6  & 219 & - ~~~~  & 66.6 $\pm$ 262 & 100.0\% & 20.5 $\pm$ 36.1 & - ~~~~  \\
    LSCB      & 10 & 244 & 111.5\% & 57.5 $\pm$ 215 & 86.43\% & 19.5 $\pm$ 32.6 & 95.06\% \\
    CB        & 4  & 182 & 83.45\% & 66.3 $\pm$ 299 & 99.66\% & 20.6 $\pm$ 37.7 & 100.4\% \\
    \midrule
    GB        & 13 & 331 & 151.2\% & 96.7 $\pm$ 347 & 145.3\% & 15.0 $\pm$ 25.0 & 72.96\% \\
    GB+BB     & 11 & 303 & 138.6\% & 83.6 $\pm$ 349 & 125.5\% & 12.8 $\pm$ 22.4 & 62.55\% \\
    GB+ECM    & 17 & 390 & 178.2\% & 91.3 $\pm$ 420 & 137.2\% & 14.6 $\pm$ 27.7 & 71.13\% \\
    GB+ECM+BB & 15 & 350 & 159.7\% & 91.2 $\pm$ 419 & 137.0\% & 14.0 $\pm$ 27.4 & 68.11\% \\
    GB+LCM    & 16 & 423 & 193.2\% & 137~ $\pm$ 541 & 206.8\% & 17.5 $\pm$ 31.3 & 85.05\% \\
    GB+LCM+BB & 14 & 349 & 159.3\% & 98.5 $\pm$ 430 & 147.9\% & 14.2 $\pm$ 26.0 & 69.09\% \\
    \bottomrule
  \end{tabular}
  \caption{Statistics of the benchmarks described in \Cref{sec:experiments} with \textbf{restarts enabled}. The timeout column indicates the number of instances that timed out for each configuration. PAR2 indicates the penalized average of runtimes, where each timeout is counted as 2 times the timeout value (2 hours). The average times and propagations are computed excluding the 50/1000 instances for which at least one of the configurations timed out.
  Standard deviation is shown after the $\pm$ symbol. Relative values (in percentages) are computed with respect to NCB and displayed on the right of the respective columns.}
  \label{tab:results-raw-restarts-650}
\end{table}

\paragraph{Cactus plots.}
\Cref{fig:results-cactus-650,fig:results-cactus-400} show cactus plots of the number of problems solved within a given time and number of synchronizations and conflicts. \Cref{fig:results-cactus-650} displays the results for the experiments presented in \Cref{sec:experiments}, while \Cref{fig:results-cactus-400} shows the results for the experiments presented in this section.

We filtered out the redundant options for readability

\begin{figure}
  \centering
  \begin{subfigure}{0.495\textwidth}
    \includegraphics[width=\textwidth]{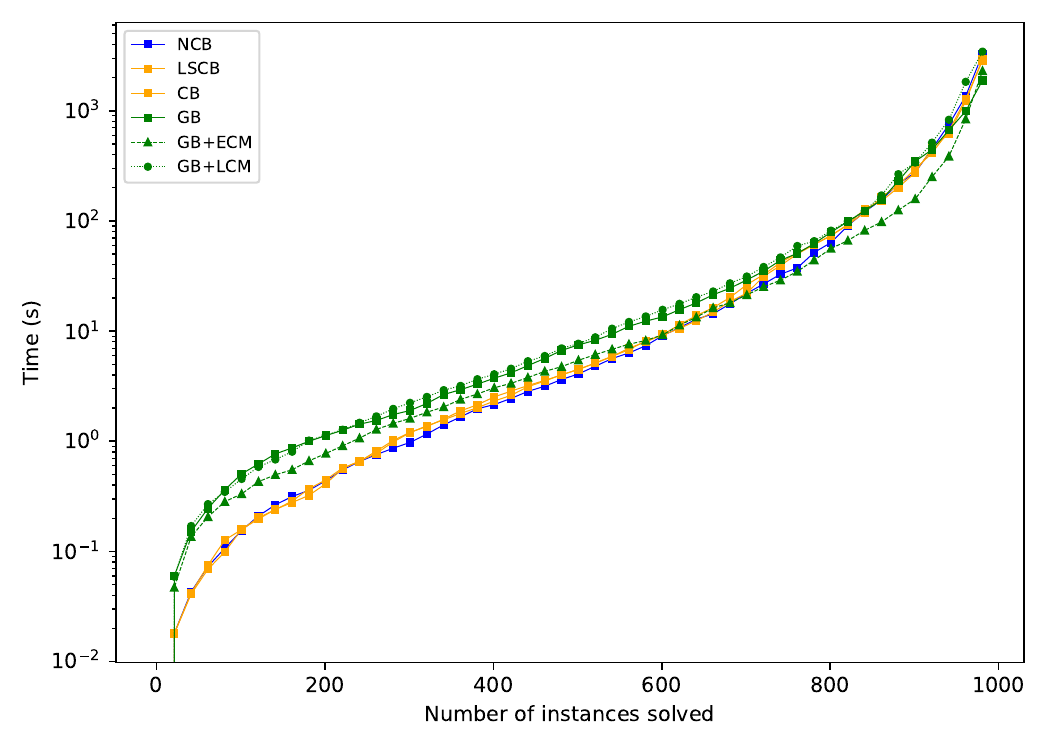}
    \caption{Number solved over time.}
    \label{fig:results-cactus-time-650}
  \end{subfigure}
  \hfill
  \begin{subfigure}{0.495\textwidth}
    \includegraphics[width=\textwidth]{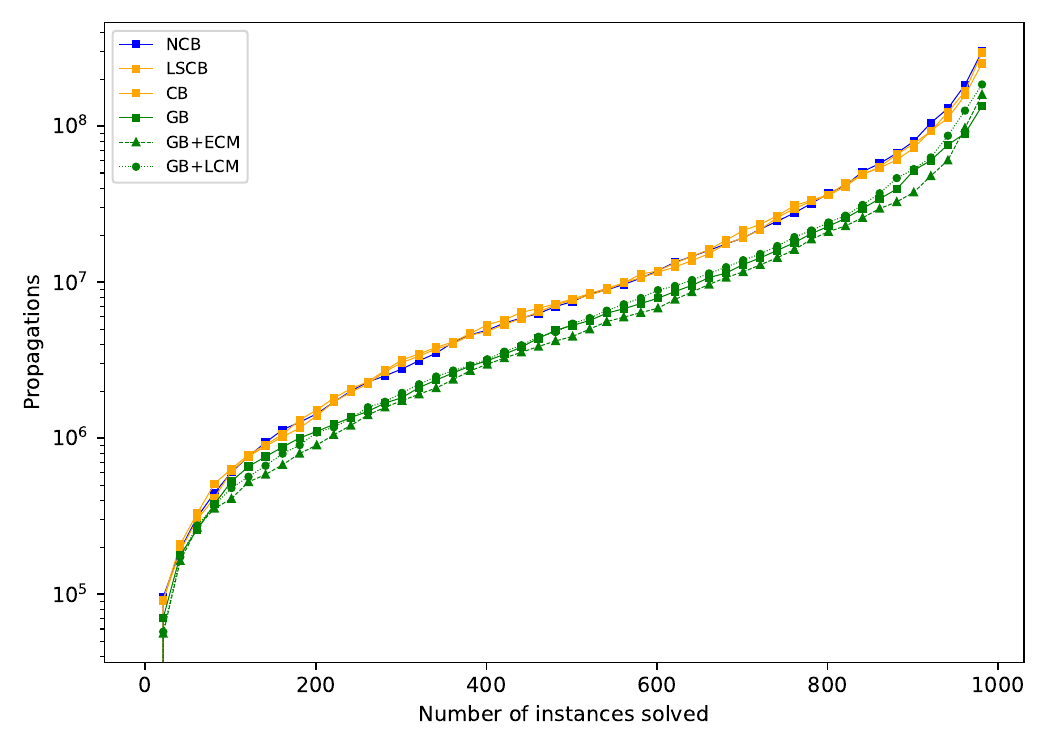}
    \caption{Number solved over number of propagations.}
    \label{fig:results-cactus-propagation-650}
  \end{subfigure}
  \begin{subfigure}{0.495\textwidth}
    \includegraphics[width=\textwidth]{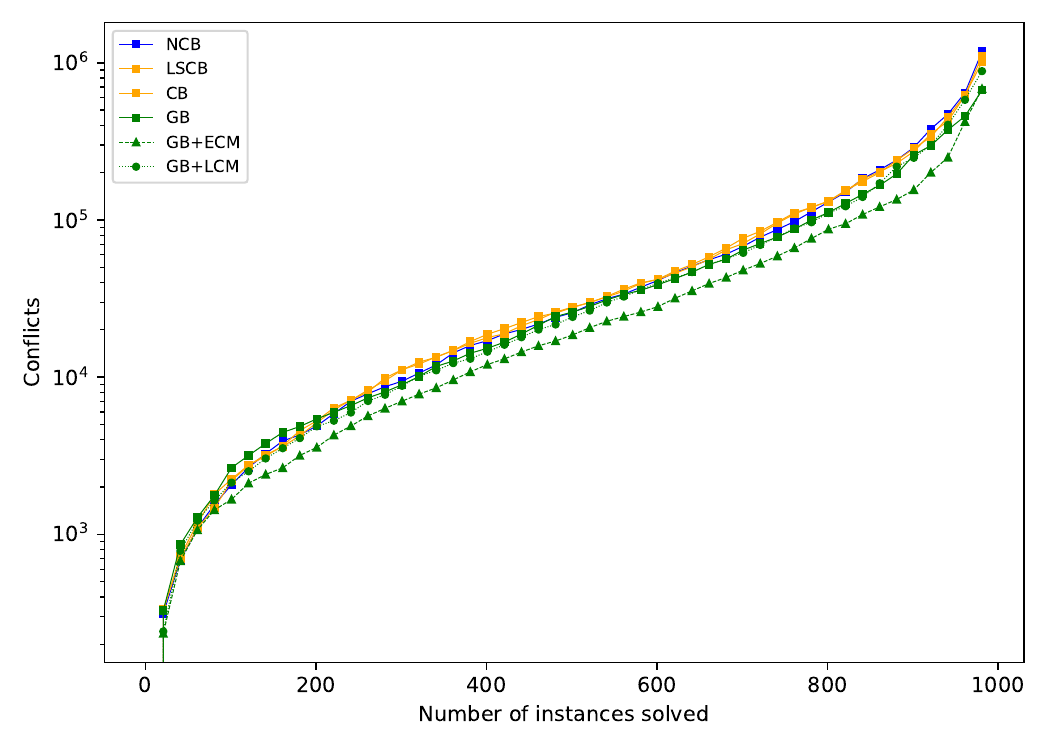}
    \caption{Number solved over number of conflicts.}
    \label{fig:results-cactus-conflicts-650}
  \end{subfigure}
  \caption{Cactus plot of the number of problems solved for NCB, CB, and GB for the experiments of \Cref{sec:experiments}. The lower the better.}
  \label{fig:results-cactus-650}
\end{figure}

\begin{figure}
  \centering
  \begin{subfigure}{0.495\textwidth}
    \includegraphics[width=\textwidth]{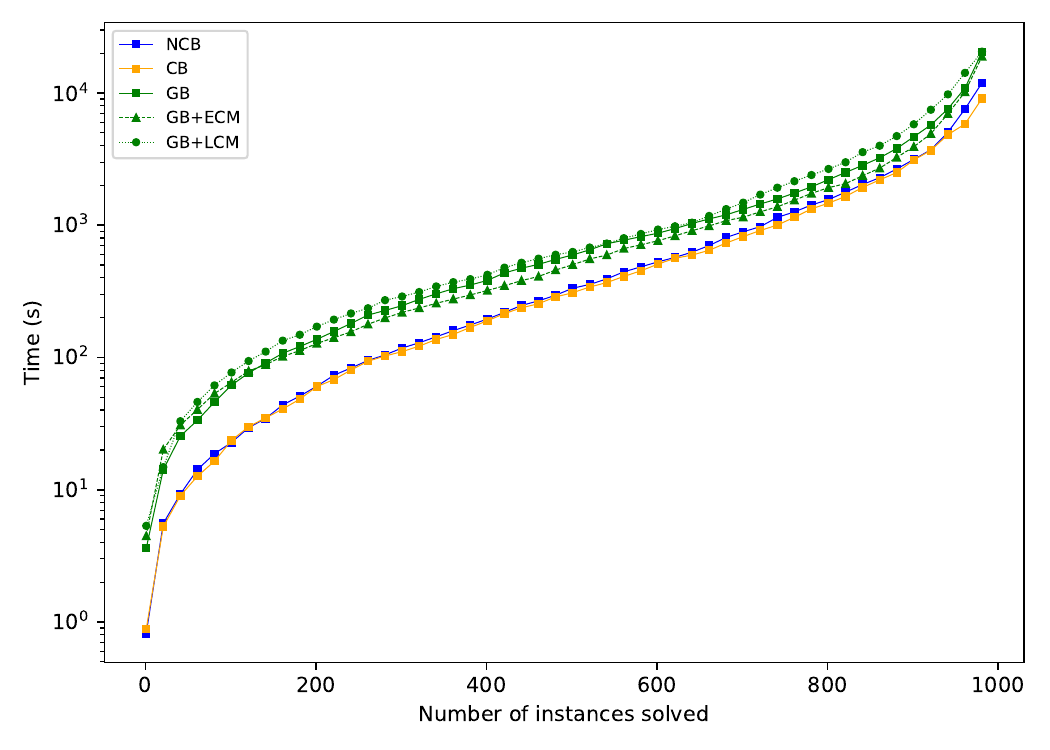}
    \caption{Number solved over time.}
    \label{fig:results-cactus-time-400}
  \end{subfigure}
  \hfill
  \begin{subfigure}{0.495\textwidth}
    \includegraphics[width=\textwidth]{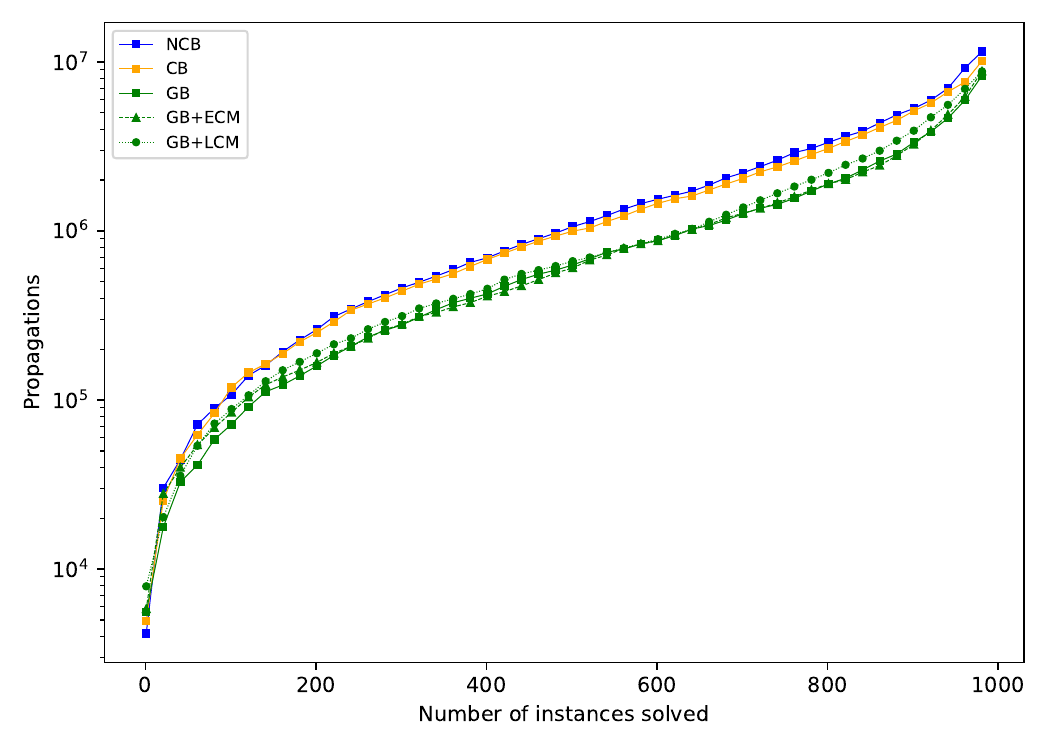}
    \caption{Number solved over number of propagations.}
    \label{fig:results-cactus-propagation-400}
  \end{subfigure}
  \begin{subfigure}{0.495\textwidth}
    \includegraphics[width=\textwidth]{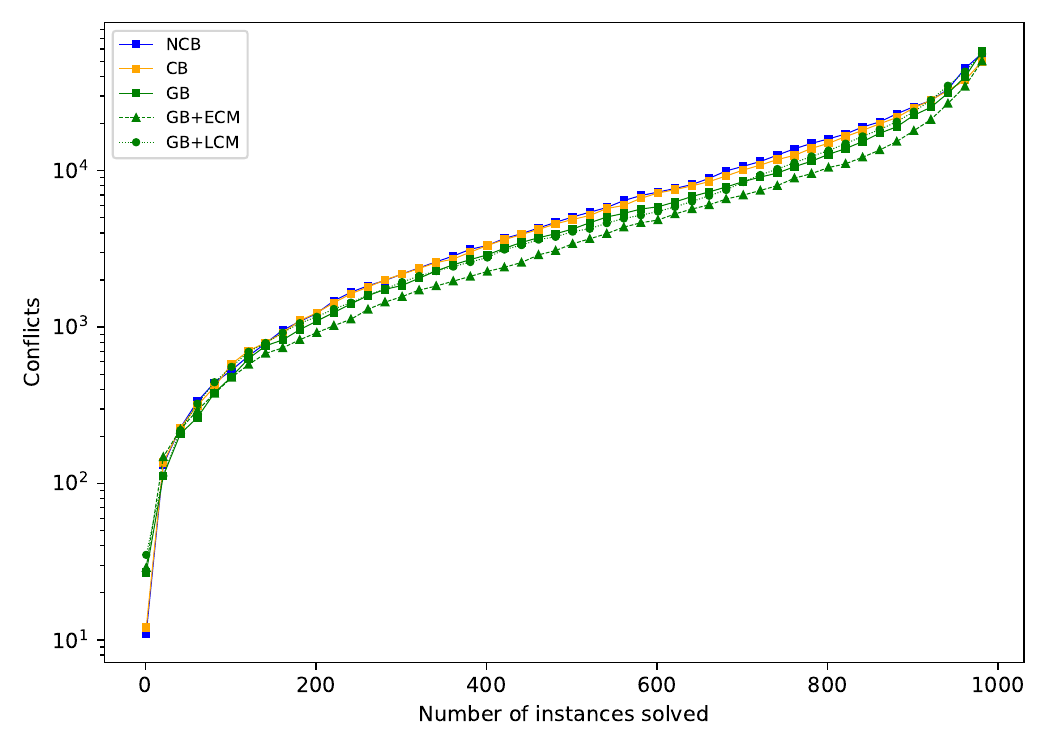}
    \caption{Number solved over number of conflicts.}
    \label{fig:results-cactus-conflicts-400}
  \end{subfigure}
  \hfill
  \begin{subfigure}{0.495\textwidth}
    \includegraphics[width=\textwidth]{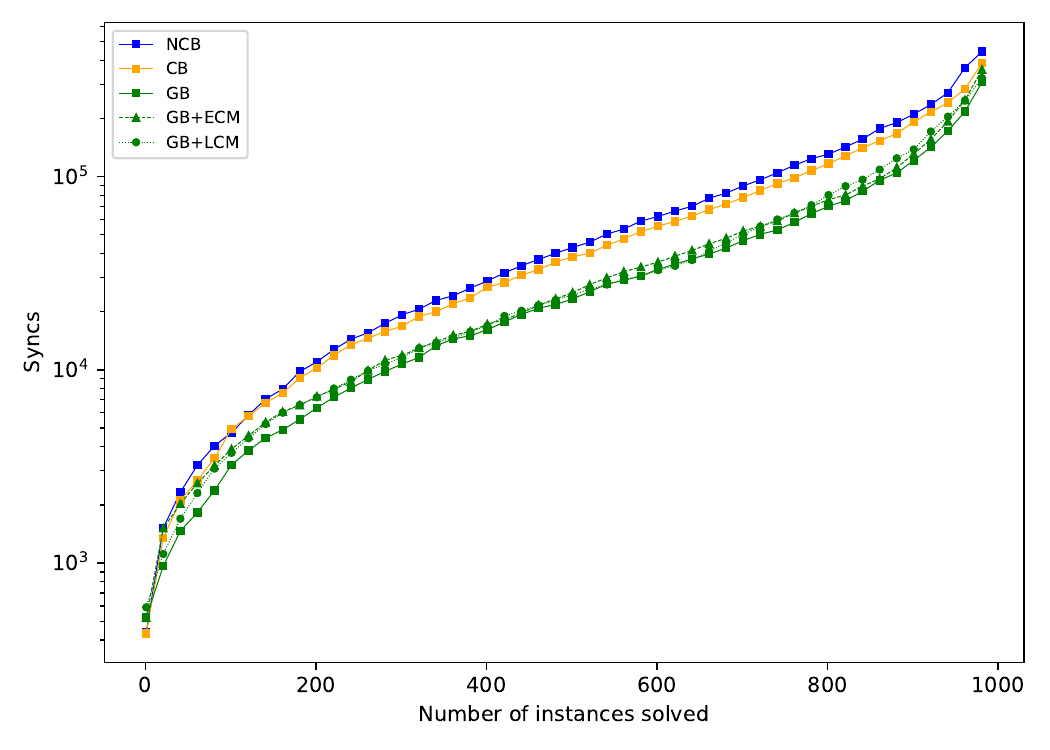}
    \caption{Number solved over number of syncs.}
    \label{fig:results-cactus-sync-400}
  \end{subfigure}
  \caption{Cactus plot of the number of problems solved for NCB, CB, and GB for the experiments of \Cref{sec:experiments}. The lower the better.}
  \label{fig:results-cactus-400}
\end{figure}

\paragraph{Conflict Search Strategies.}
The experiments show that ECR increases the computation cost of the solver by orders of magnitude, while only reducing the number of synchronizations by a few percents. ECR does benefit GB slightly more than CB and NCB, but overall, the cost of ECR seems quite prohibitive for the gain it provides.

PCR does not show any significant improvement in terms of number of synchronizations. The implementation effort of PCR being very similar to ECR, we believe that PCR is not worth the cost either.

Different conflict search strategies could be useful in applications with very high cost of synchronizations, where the overhead of ECR and PCR could be amortized. However, there probably are more efficient ways to improve performance of such applications.

We conclude that PCR/ECR are not worth the cost.

\paragraph{Chunk Merging.}
Interestingly, on pure SAT problems, ECM seems to be the most effective chunk merging strategy for computation time, while no merging seems slighlty better for the number of synchronizations. As indicated in \Cref{sec:chunk-merging}, ECM makes backtracking more aggressive, which can lead to additional unassigned literals and therefore more synchronizations. On the other hand, LCM is more conservative but seems to have a detrimental effect on both the solving time and the number of synchronization. The solving time penalty probably comes from the fact that LCM makes chunk cost calculation more expensive. The reason for the increase in synchronizations is less clear. It could be due to the quality of learned clauses using the missed implications.

\paragraph{Restarts}
Unsurprisingly, restarts do not synergize well with GB. They increase both the solving time and the number of synchronizations. This is because restarts cancel a lot of the expensive bookkeeping of GB, without exploiting the benefits of GB. In particular, there is no point in trying to preserve the trail as much as possible, if we are going to throw it away with a restart. This is also illustrated by the fact that the relative performance of GB compared to NCB is much worse with restarts than without.

\paragraph{Learning vs. cheapest chunk.}
The option BB can sometimes have a negative effect on the performance of the solver. This might be due to the fact that it introduces clauses that are not really relevant to the conflict, which can lead to worse behavior after conflict repair.
Interestingly, BB does seem to behave better with the larger problems ran in \Cref{sec:experiments}, which suggests that it might be more effective when the cost of synchronizations is higher.

\paragraph{Cactus plot analysis.}
The cactus plots in \Cref{fig:results-cactus-650} show that easy problems are solved faster with (N)CB, while GB seems to scale better on harder problems. This is probably because GB has a higher overhead than (N)CB, but saves more propagations on harder problem instances.

The cactus plots in \Cref{fig:results-cactus-400} show that GB consistently executes fewer synchronizations than (N)CB.

In both experiments, the number of conflicts remained relatively similar between the different configurations, with a slight decrease for GB. This might be particular to the category of problems we are looking at, where a more local search is effective.



\section{Analysis on SMT conflicts}
\label{sec:app:case-study}
In addition to analysis on graph coloring benchmarks, we generated conflicts from SMT solver to evaluate how different backtracking mechanisms are performing in this domain.
The option acronyms were defined in \Cref{sec:experiments}.

\paragraph{Experiment setup.}
We patched the SMT solver \cvc~\cite{cvc5} to export every conflict within its default propositional solver (\minisat~\cite{DBLP:conf/sat/EenS03}).
For each conflict we logged all current clauses as well as the trail leading to the conflict.
We collected over 2.9 million conflicts by executing \cvc on all benchmarks in the QF\_UF and QF\_NIA suite of the SMT-LIB release 2024~\cite{smtlib2024} with a timeout of 15 seconds per instance.

To reconstruct the conflicts in \napsat, we converted each logged conflict to an input in the DIMACS CNF format, added unit clauses for assumptions, and instructed the solver to take the same decisions as within \cvc.

\paragraph{Results.}
Compared to the previous experiments we need to adapt our measurement metric slightly and count the number of backtracked literals instead of synchronizations.
This is because we cannot track subsequent literal synchronizations after a conflict as any measurement beyond the immediate conflict handling requires the SMT solver in the loop.
Furthermore, we utilize a constant cost function in order to optimize for the number of backtracked literals.

We start each measurement at the logged conflict and measure until the next SAT solver decision is to be made.
Note that a conflict analysis propagates literals which may lead to subsequent conflicts.
Those are included in our measurements as this closer reflects the behavior in an SMT solver.
To make average values more meaningful, we have clustered the conflicts by the sum $s$ of total backtracked literals for both methods:
\begin{center}
\begin{tabular}{ll}
tiny & ($s \leq 10$)\\
small & ($10 < s \leq 50$)\\
medium & ($50 < s \leq 100$)\\
big & ($100 < s \leq 500$)\\
huge & ($500 < s$)
\end{tabular}
\end{center}
\Cref{tab:results-smt} compares GB with NCB on the average number of backtracked literals (literals of the reconstructed trail that were unassigned during a conflict analysis), the average number of conflicts (the logged conflict plus subsequent conflicts), as well as the percentage of instances where one approach surpasses the other.

\newcommand{\x}{\(\pm\)\xspace}
\begin{table}[t]
	\centering
	\begin{tabular}{l|P{1.77cm}|P{1.77cm}|P{1.77cm}|P{1.77cm}|P{1.77cm}}
		\toprule
		& tiny & small & medium & big & huge \\
		\midrule
		\textbf{\#Instances} & 67,274 & 411,097 & 339,571 & 1,847,692 & 243,486 \\
		\midrule
		\textbf{Ratios \#BT} &&&&& \\
		\hspace{1em}NCB > GB & 46.81\% & 49.53\% & 52.20\% & 48.95\% & 58.38\% \\
		\hspace{1em}NCB = GB & 44.14\% & 33.94\% & 24.39\% & 20.72\% & 16.83\% \\
		\hspace{1em}NCB < GB &  7.70\% & 16.52\% & 23.41\% & 30.33\% & 24.79\% \\
		\midrule
		\textbf{Avg. \#BT} &&&&& \\
		\hspace{1em}NCB & 4.53 \x 1.42 & 16.16 \x 7.85 & 42.1 \x 12.7 & 139.0 \x 57.8 & 406 \x 400 \\
		\hspace{1em}GB  & 3.35 \x 1.33 & 12.18 \x 6.45 & 32.6 \x 12.1 & 127.6 \x 56.6 & 300 \x 206 \\
		\midrule
		\textbf{Avg. \#Conflicts} &&&&& \\
		\hspace{1em}NCB & 1.01 \x 0.09 & 1.07 \x 0.32 & 1.20 \x 0.51 & 1.56 \x 0.83 & 1.83 \x 1.06\\
		\hspace{1em}GB  & 1.10 \x 0.34 & 1.35 \x 0.82 & 1.68 \x 1.21 & 2.18 \x 1.56 & 2.62 \x 2.96\\
		\bottomrule

	\end{tabular}

	\caption{Statistics on backtracked literals (BT) and conflicts.
		Standard deviation is shown after the $\pm$ symbol.}
	\label{tab:results-smt}
\end{table}

\paragraph{Analysis.}
Performing GB on conflicts on SMT solver generated conflicts performs fewer literal unassignments than NCB in most cases.
Subsequent conflicts are, however, more likely to appear with GB.
Analyzing a single conflict, GB performs fewer or equal backtracking operations than NCB by design. However, as GB and NCB backtrack different literals, they might encounter different secondary conflicts. In particular, GB encounters more of them.
Intuitively, this is quite natural. Since GB preserves more of the trail, it is more likely to trigger additional implications that can lead to further conflicts. An extreme example is when NCB backtracks all the way to the root level. It could be that the trail now only contains a single literal, making further conflicts without decisions very unlikely.
In this case, GB learns more clauses.

\end{document}